\documentclass[a4paper,runningheads]{lmcs}
\pdfoutput=1

\usepackage{lastpage}
\lmcsdoi{15}{1}{8}
\lmcsheading{}{\pageref{LastPage}}{}{}%
{Jan.~03,~2018}{Feb.~05,~2019}{}

\usepackage[algoruled,vlined]{algorithm2e}
\SetKwProg{Fn}{Function}{}{}

\usepackage{amsmath}
\usepackage{mathabx}
\usepackage{relsize}
\usepackage{xcolor}
\usepackage{float}
\usepackage{amsmath,amssymb}
\usepackage{mathpartir}
\usepackage{etoolbox}
\usepackage{empheq,soul}
\usepackage{upgreek}
\usepackage{hyperref}
\usepackage{multicol}
\usepackage{hyperref}
\usepackage[capitalise]{cleveref}

\usepackage{framed}

\usepackage{mathtools}
\numberwithin{equation}{section}

\setstcolor{red}

\usepackage{bussproofs}

\usepackage{enumitem}

\usepackage[curve,matrix,arrow,cmtip]{xy}
\usepackage[T1]{fontenc}
\usepackage[utf8]{inputenc}

\usepackage{hyperref}

\def \bcompaula {\begin{color}{red}} 
 \def \ecompaula {\end{color}}

\newcommand{\recn}[1]{\mathtt{r}_{#1}}

\newcommand{\NP}{{\bf NP}}

\newcommand{\B}{\mathit{itob}}

\newcommand{\transfKB}{ \mathit{clear}}
\newcommand{\TransfKB}[1]{ \transfKB (#1) }
\newcommand{\wtei}[4]{#1 \vdash #2 :_{#4} #3}

\newcommand{\comma}{, }

\newif\ifproofs
\proofstrue

\newif\ifcomments
\commentstrue

\newif\iflong
\longtrue

\newcommand{\mycase}[1]{{\bf #1 }.}

\definecolor{mymagenta}{rgb}{0.5,0,0.5}
\definecolor{myred}{rgb}{0.5,0,0}
\definecolor{mygreen}{rgb}{0,0.4,0}
\definecolor{myblue}{rgb}{0,0,0.6}

\newcommand{\set}[1]{\{#1\}}
\newcommand{\natset}{\mathbb{N}}

\newcommand{\mkkeyword}[1]{\mathtt{\color{myblue}#1}}
\newcommand{\mkconstant}[1]{\mathtt{\color{mymagenta}#1}}

\newcommand{\ttcomma}{\texttt{\upshape,}}

\newcommand{\defrule}[1]{%
  \hypertarget{rule:{#1}}{%
    \text{\scriptsize[\textsc{#1}]}%
  }%
}

\newcommand{\refrule}[1]{%
  \hyperlink{rule:{#1}}{%
    \text{\upshape\scriptsize[\textsc{#1}]}%
  }%
}

\newcommand{\lambdabb}{\lambdab}
\newcommand{\lambdabbminus}{(\lambdab)^-}

\newcommand{\modalone}{\bullet}
\newcommand{\lambdab}{\lambda^{\modalone}_{\rightarrow}}

\newcommand{\lambdaKB}{\lambda^{\text KB}}
\newcommand{\lambdaNakano}{\lambda_{\text Nak}}

\newcommand{\fromtypeNak}{\mathit{tree}}

\newcommand{\toB}[1]{\fromtypeNak(#1)}

\newcommand{\munf}{\mathit{nf}_{\!\!\!\mu}}
\newcommand{\Munf}[1]{\munf (#1)}

\newcommand{\introarrow}{$\arrow$I}
\newcommand{\elimarrow}{$\arrow$E}
\newcommand{\introbullet}{$\modalone$I}

\newcommand{\treeEq}{=}

\newcommand{\axiom}{axiom}
\newcommand{\const}{const}
\newcommand{\prodconst}{$\times$ \const}

\newcommand{\natzero}{$\nat \zero$}
\newcommand{\natsucc}{$\nat \suc$}

\newcommand{\Constant}{\mkconstant{k}}

\newcommand{\var}{\varX}
\newcommand{\varX}{{x}}
\newcommand{\varY}{{y}}
\newcommand{\varZ}{{z}}

\newcommand{\Expression}{\ExpressionE}
\newcommand{\ExpressionE}{e}
\newcommand{\ExpressionF}{f}

\newcommand{\Type}{\TypeT}
\newcommand{\TypeT}{t}
\newcommand{\TypeS}{s}

\newcommand{\llt}{\mathit{LLT}}
\newcommand{\LLT}[1]{\llt (#1)}
\newcommand{\bt}{\mathit{BT}}
\newcommand{\BT}[1]{\bt(#1)}

\newcommand{\hnf}{\mathit{hnf}}
\newcommand{\Hnf}[1]{\hnf(#1)}

\newcommand{\TypeN}{\TypeNT}
\newcommand{\TypeNT}{T}
\newcommand{\TypeNS}{S}
\newcommand{\ContextNak}{\Sigma}

\newcommand{\Context}{\mathrm{E}}
\newcommand{\PContext}{\mathrm{C}}

\newcommand{\Hole}{[~]}

\newcommand{\MetaType}{T}
\newcommand{\MetaTypeS}{S}

\newcommand{\TypeVar}{X}
\newcommand{\TypeVarY}{Y}

\newcommand{\TypeVarA}{X_1}

\newcommand{\TypeVarC}{X_3}
\newcommand{\TypeVarD}{X_4}

\newcommand{\substT}{\tau}
\newcommand{\substTE}{\sigma}
\newcommand{\substE}{\substR}

\newcommand{\substR}{\theta}

\newcommand{\update}[3]{#1 \{ #2 \mapsto  #3 \}}
\newcommand{\singleTS}[2]{\{#1 \mapsto #2 \}}
\newcommand{\newTSubs}[2]{#1 \mapsto #2}

\newcommand{\appTS}[2] { #1 (#2)}
\newcommand{\mapTS}[4]{\{ #1 \mapsto  #2, \ldots,  #3 \mapsto #4 \}}

\newcommand{\IntVar}{N}
\newcommand{\IntVarM}{M}
\newcommand{\IntVara}{N_1}
\newcommand{\IntVarc}{N_2}
\newcommand{\IntVard}{N_3}
\newcommand{\IntVare}{N_4}

\newcommand{\IntExp}{E}

\newcommand{\setC}{\mathcal{C}}

\newcommand{\eqc}[1]{{#1}^{=}}

\newcommand{\intc}[1]{#1^{\mathbb{N}}}

\newcommand{\setD}{\mathcal{V}}

\newcommand{\setCone}{\setC_1}
\newcommand{\setCtwo}{\setC_2}

\newcommand{\setE}{\mathcal{E}}

\newcommand{\setT}{\mathcal{T}}

\newcommand{\typec}[4] {#1 \vdash #2 : #3 \ \mid \  #4}

\newcommand{\equal}{  \stackrel{{ ?}}{{=}}  }

\newcommand{\mayor}{ \ \stackrel{{ ?}}{{\geq} \ }}

\newcommand{\menor}{\ \stackrel{{ ?}}{{<}} \ }

\newcommand{\binaryconstructor}{ \; \mathtt{ op} \;}

\newcommand{\cond}{p}

\newcommand{\assign}[2]{#1 \leftarrow #2}

\newcommand{\unify}{\mathtt{ unify}}
\newcommand{\Unify}[2]{\unify (#1, #2)}
 \newcommand{\isNatOpBox}[1]{\mathtt{isNatorOp?}(#1)}

\newcommand{\equalcons}[2]{\mathtt{EqualCons?}(#1, #2)}

\newcommand{\pone}{\mathtt{True}}
\newcommand{\ptwo}{ \mathtt{Diff_{\infinite}}?}
\newcommand{\pthree}{\infinite{\mbox{\!-}}\mathtt{free}?}
\newcommand{\pfour}{\mathtt{tailfinite?}}

\newcommand{\truecond}[1]{\pone (#1)}
\newcommand{\equalinfinite}[1]{ \ptwo (#1)}
\newcommand{\freeinfinite}[1]{\pthree (#1)}
\newcommand{\tailfinite}[1]{\pfour (#1)}

\newcommand{\guards}{\mathtt{guard}}
\newcommand{\Guards}[1]{\guards(#1)}

\newcommand{\countb}{\mathtt{count}}
\newcommand{\Countb}[2]{\countb(#1,#2)}

\newcommand{\finitealt}{\mathtt{arrows}}
\newcommand{\Finitealt}[2]{\finitealt(#1,#2)}

\newcommand{\infertype}{\mathtt{ infer}}
   
 \newcommand{\allconstraints}{\mathtt{ SubC}}
  \newcommand{\subtypes}{\mathtt{ SubT}}

 \newcommand{\subtrees}{\mathtt{trees}}
 \newcommand{\psubtrees}{\mathtt{ptrees}}
 
  \newcommand{\Subtrees}[1]{\subtrees(#1)}
 \newcommand{\Psubtrees}[1]{\psubtrees(#1)}
 
 \newcommand{\size}[1]{\mid \! \! #1 \! \! \mid } 

\newcommand{\unit}{\mkconstant{unit}}

\newcommand{\pair}{\mkconstant{pair}}
\newcommand{\fst}{\mkconstant{fst}}
\newcommand{\snd}{\mkconstant{snd}}
\newcommand{\fix}{\mkconstant{fix}}

\newcommand{\await}{\mkconstant{await}}
\newcommand{\Await}[1]{\await~#1}

\newcommand{\natrec}{\mkconstant{natrec}}
\newcommand{\listrec}{\mkconstant{lrec}}
\newcommand{\consl}{\mkconstant{consl}}
\newcommand{\nil}{\mkconstant{nil}}

\newcommand{\case}{\mkconstant{case}}
\newcommand{\inl}{\mkconstant{inl}}
\newcommand{\inr}{\mkconstant{inr}}

\newcommand{\Unit}{\mkbasictype{Unit}}

\newcommand{\ifkw}{\mkkeyword{if}}
\newcommand{\thenkw}{\mkkeyword{then}}
\newcommand{\elsekw}{\mkkeyword{else}}

\newcommand{\Fun}[1]{\lambda#1.}

\newcommand{\Pair}[2]{\langle #1\ttcomma#2\rangle}
\newcommand{\Fst}[1]{(\fst~#1)}
\newcommand{\Snd}[1]{(\snd~#1)}

\newcommand{\Ebullet}[1]{\modalone #1}

\newcommand{\skipeven}{\mkconstant{skip}}

\newcommand{\foldr}{\mkconstant{foldr}}

\newcommand{\mkbasictype}[1]{\mkkeyword{#1}}

\newcommand{\arrow}{\to}
\newcommand{\tbullet}[1][]{%
  \modalone%
  \ifblank{#1}{}{^{#1}}%
}

\newcommand{\subst}[2]{ [#1 /#2 ]}

\newcommand{\dom}{{dom}}

\newcommand{\teq}{=}

\newcommand{\nakteq}{\simeq}

\newcommand{\wte}[3]{#1 \vdash #2 : #3}

\newcommand{\red}{\longrightarrow}

\newcommand{\TypeContext}{\Upgamma}

\newcommand{\TypeContextD}{\Delta}

\newcommand{\ConsCtx}{\kappa}
\newcommand{\typeOf}[3]{\mathtt{typeOf}(#1,#2,#3)}

\newcommand\ldana{(\![}
\newcommand\rdana{]\!)}
\newcommand\ti[1]{{\ldana{#1}\rdana}}
\newcommand\indti[2]{{\ldana{#1}\rdana}_{#2}}

\newcommand{\typerank}{\mathit{rank}}

\newcommand{\ind}{i}
\newcommand{\indj}{j}
\newcommand{\funsubst}{\rho}
\newcommand{\EE}{\mathbb{E}}

\newcommand{\setWN}{\mathbb{WN}}
\newcommand{\setNVAR}{\mathbb{HV}}
\newcommand{\HH}{\setWN}
\newcommand{\WNVAR}{\setNVAR }

\newcommand{\modelsi}{\models_{\ind}}
\newcommand{\modelsj}{\models_{\indj}}

\newcommand{\infinite}{\mkbasictype{{\modalone^\infty}}}

\newcommand{\ListN}{\List{\nat}}
\newcommand{\ListNtwo}{\Listtwo{\nat}}

\newcommand{\badListN}{\badList{\nat}}

\newcommand{\List}[1]{\mkbasictype{Str1}_{#1}}
\newcommand{\Listtwo}[1]{\mkbasictype{Str2}_{#1}}

\newcommand{\badList}[1]{\mkbasictype{Str}'_{\nat}}
\newcommand{\Listhalf}[1]{\mkbasictype{StrHalf}_{#1}}

\newcommand{\Lista}[1]{\mkbasictype{List}_{#1}}

\newcommand{\Etype}{\mkbasictype{E}}
\newcommand{\dtwo}{\mkbasictype{F}}
\newcommand{\dthree}{\mkbasictype{G}}
\newcommand{\coList}[1]{\mkbasictype{coList}_{#1}}

\newcommand{\addlist}{\mkconstant{sum}}
\newcommand{\eqels}{\mkconstant{eq?}}

\newcommand{\coNat}{\mkbasictype{coNat}}
\newcommand{\cozero}{\mkconstant{z}}
\newcommand{\cosucc}{\mkconstant{s}}
\newcommand{\add}{\mkconstant{add}}

\newcommand{\minus}{\mkconstant{minus}}

\newcommand{\fixtype}{\mkbasictype{D}}

\newcommand{\nat}{\mkbasictype{Nat}}
\newcommand{\suc}{\succesor \ }
\newcommand{\succesor}{\mkconstant{succ}}
\newcommand{\Suc}[1]{(\succesor~#1)}

\newcommand{\bool}{\mkbasictype{Bool}}
\newcommand{\stream}[2]{\mkbasictype{Stream} ({#1}, {#2})}

\newcommand{\interleave}{\mkconstant{interleave}}
\newcommand{\map}{\mkconstant{map}}
\newcommand{\maap}{\mkconstant{maap}}
\newcommand{\nats}{\mkconstant{nats}}
\newcommand{\suces}{\suc}
\newcommand{\naats}{\mkconstant{naats}}
\newcommand{\merge}{\mkconstant{merge}}
\newcommand{\hamming}{\mkconstant{ham}}
\newcommand{\pairup}{\mkconstant{pairup}}
\newcommand{\filter}{\mkconstant{filter}}
\newcommand{\ones}{\mkconstant{ones}}

\newcommand{\zero}{\mkconstant{0}}

\newcommand{\iterate}{\mkconstant{iterate}}
\newcommand{\fib}{\mkconstant{fib}}

\newcommand{\get}{\mkconstant{get}}
\newcommand{\take}{\mkconstant{take}}

\newcommand{\idfun}{\mkconstant{I}}
\newcommand{\omegaterm}{\mkconstant{\Upomega}}

\newcommand{\app}[2]{\mbox{$(#1 \ #2)$}}

 \theoremstyle{definition}
 \newtheorem{definition}{Definition}[section]

 \newtheorem{lemma}[definition]{Lemma}
 \newtheorem{theorem}[definition]{Theorem}
 \newtheorem{corollary}[definition]{Corollary}

 \newcommand{\ri}[2]{{#2} of {#1}}
\newcommand{\setcat}{{Set}}

\newcommand{\topos}{\mathcal{S}}

\newcommand{\sem}[1]{[\! \![#1] \! \!]}

\newcommand{\Next}{\blacktriangleright}

\newcommand{\isomarrow}{\xi}
\newcommand{\isomprod}{\theta}

\newcommand{\next}{{next}}
\newcommand{\nextn}[1]{{next}^{#1}}

\newcommand{\incFunctor}{\mathcal{I}}
\newcommand{\IncFunctor}[1]{\incFunctor (#1)}

\newcommand{\isomarrown}[1]{\xi^{#1}}
\newcommand{\isomprodn}[1]{\theta^{#1}}
\newcommand{\eval}{{eval}}
\newcommand{\curry}{{curry}}

\newcommand{\expon}[2]{{#2}^{#1}}

\newcommand{\oeqctx}{\approx}

\begin{document}

\title[A Light Modality for Recursion]{A Light Modality for Recursion}

\author[Paula Severi]{Paula Severi}
\address{Department of Computer Science, University of Leicester, UK}



  \begin{abstract}
    We investigate a modality for controlling 
     the behaviour of recursive functional 
     programs on infinite structures 
     which is completely silent in the syntax.
     The latter means that 
     programs do not contain  ``marks'' showing 
     the application of the introduction and elimination rules for the modality.
     This shifts the burden of controlling recursion
    from the programmer  to the compiler.  
    To do this, we introduce a typed lambda calculus  \`a la Curry
    with a
     silent modality  and  guarded 
    recursive types. 
    The  typing discipline      guarantees  normalisation and can be
    transformed into an algorithm which infers the type of a program.
    \end{abstract}

\maketitle

\section{Introduction}
\label{section:Introduction}

The quest of finding a typing discipline that 
guarantees that functions on coinductive data types
are productive has prompted a variety of works
that rely on a modal operator  \cite{Nakano00:lics,KrishnaswamiB11,severidevriesICFP2012,DBLP:conf/popl/CaveFPP14,AM13,CBGB15,DBLP:conf/fossacs/BizjakGCMB16,Guatto2018}.
Typability in these systems guarantees normalization and non-normalizing  programs
such as $\fix~\idfun$ where $\idfun = \lambda x. x$ are not typable.
All these works  
except for Nakano's~\cite{Nakano00:lics}
have explicit constructors and destructors  in the syntax
of programs  
~\cite{KrishnaswamiB11,severidevriesICFP2012,DBLP:conf/popl/CaveFPP14,CBGB15,Guatto2018}.  
This has the advantage that type checking and/or type inference are
easy but  it has the disadvantage that
they do not really liberate 
the programmer from the task of controlling recursion
since one has 
to know when to apply the introduction and elimination rules
for the modal operator.

We present  a  typed lambda calculus with a temporal
modal operator $\modalone$  called $\lambdab$ 
which 
has the  advantage of having the modal operator silent in 
the syntax of programs
 The system $\lambdab$ does not  need to make use of 
  a subtyping relation as Nakano's.
Even without a subtyping relation, 
 the modal operator is still challenging to deal with because types are intrinsically
non-structural, not corresponding to any expression form in the
calculus.

 The modal operator $\modalone$, also called {\em the delay operator},
 indicates that the data on the recursive calls
will only be available later.
Apart from the modal operator, we also include  guarded 
recursive types 
which generalize the  recursive equation 
$\List{\Type} = \Type \times \bullet \List{\Type}$
for streams 
\cite{KrishnaswamiB11,severidevriesICFP2012}.
This allows us to type  productive functions 
on streams such as
 \[
  \skipeven~xs = 
  \Pair{\fst~xs}{\skipeven~(\snd~(\snd~xs))}
  \]
  which deletes the elements at even positions of a stream
  using the  type $\ListN \arrow \ListNtwo$
  where 
  $\Listtwo{\Type} = \Type \times \bullet \bullet \Listtwo{\Type}$.
  The temporal modal operator $\modalone$ allows to type many productive functions 
which are rejected by  the proof assistant Coq
\cite{Coquand93}.

Lazy functional programming is acknowledged as a paradigm that fosters
software modularity~\cite{Hughes89} and enables programmers to specify
computations over possibly infinite data structures in elegant and
concise ways.    
   We give some examples that show how modularization 
   and compositionality  can be
   achieved using the modal operator.
   An important recursive pattern used in functional programming for modularisation
  is  $\foldr$ 
   defined by 
  \[
  \foldr~f~xs = f~(\fst~xs)~(\foldr~f~(\snd~xs))
  \]
  The  type of $\foldr$ which is 
  $(\Type \arrow \bullet \TypeS  \arrow \TypeS) 
   \arrow \List{\Type} \arrow \TypeS$, is telling us
   what is the safe way to build functions.
   While it is possible to type 
   \[
  \iterate~f= \foldr~(\Fun{x y}{\Pair{f~x}{y} }) 
\] 
by assigning the type $\Type \arrow \Type$ to $f$,   and assuming $\TypeS=\List{\Type}$, 
it is not possible to type the unproductive function 
  $\foldr~(\Fun{x y}{y})$.
   This is because  $\Fun{x y}{y}$ does not have type
  $ (\Type \arrow \bullet \TypeS  \arrow \TypeS)$.

In spite of the fact that $\ListN \arrow \ListNtwo$
    is the type of $\skipeven$ with the minimal number of bullets,
     it does not give  enough information to  
   the programmer to know that the composition of $\skipeven$
   with itself  is still typeable. 
   The type inference algorithm for  $\lambdab$   can 
    infer a more general type for $\skipeven$, 
    which looks like
   \[
   \stream{\IntVar}{\TypeVar}
    \arrow \stream{\IntVar+1}{\TypeVar}
   \]  
   where 
   $\stream{\IntVar}{\TypeVar} =
   \TypeVar \times \bullet^{\IntVar} \stream{\IntVar}{\TypeVar}
   $
   and 
   $\IntVar$ and $\TypeVar$
   are integer and type variable, respectively.
   It is clear now that from
    the above type, a programmer can deduce  that
   it is possible to do  
   the composition of $\skipeven$ with itself.

\subsection{Outline}
\cref{section:calculus} defines  the 
 typed lambda calculus $\lambdab$ with the silent modal operator
 and shows some examples.
Section \ref{section:properties} proves subject reduction and normalisation.
It also shows  that the  L\`evy-Longo and B\"ohm trees of a typeable expression 
have no $\bot$
for $\infinite$-free and tail finite types respectively.
Section \ref{section:comparision} 
compares $\lambdab$ with the type systems of
\cite{Nakano00:lics} and \cite{KrishnaswamiB11}.
Section \ref{section:semantics} gives  an adequate denotational semantics
for $\lambdab$ in the topos of trees as a way of linking our system to the work 
by Birkedal et al \cite{LMCSBirkedaletal}.
Section \ref{section:typeinference} shows 
 decidability of the type inference problem for $\lambdab$.
This problem is solved by an algorithm 
which  has the interesting feature of 
combining unification of types with 
integer linear programming.
Sections~\ref{section:relatedwork} and~\ref{section:conclusions}
 discuss related and future work, respectively. 
 
 \subsection{Contribution}
This paper improves and  extends  \cite{DBLP:conf/fossacs/Severi17}  in several ways:
\begin{enumerate}

\item We prove that the typed lambda calculus of Krishnaswami and Benton \cite{KrishnaswamiB11}
is included in $\lambdab$ and that $\lambdab$ is included in the one
of Nakano \cite{Nakano00:lics}.

\item We give conditions on the types that guarantee that 
 the L\`evy-Longo and B\"ohm trees of a typable expression have no $\bot$.

\item Besides   proving soundness for the denotational semantics, 
we also prove that it is adequate with respect to an observational equivalence.

\item We obtain  automatic ways of guaranteeing that
the program satisfies the properties of  normalization,  
  having L\`evy-Longo and   B\"ohm tree without $\bot$. 

\end{enumerate}

\section{The Typed Lambda Calculus with the Delay Modality}
\label{section:calculus}

This section defines the type discipline for $\lambdab$.
Following~\cite{Nakano00:lics}, we introduce the \emph{delay} type
constructor $\tbullet$, so that an expression of type $\tbullet\Type$
denotes a value of type $\Type$ that is available ``at the next moment
in time''.
This constructor is key to control recursion and attain normalisation
of expressions.

\subsection{Syntax}
\label{section:syntax}

The syntax for   {\em expressions}  and 
{\em types} 
is given by the following grammars.
 \[
 \begin{array}{@{}c@{\quad \quad \quad}c@{}}
  \begin{array}[t]{r@{~~}c@{~~}l@{\quad}l}
    \Expression & ::=^{ind} & & \textbf{Expression} \\
    &  & \var       & \text{(variable)} \\
      & |  & \Constant & \text{(constant)} \\
    & | & \Fun\var\Expression & \text{(abstraction)} \\
    & | & \Expression\Expression & \text{(application)} 
%
  \end{array}
  &
  \begin{array}[t]{r@{~~}c@{~~}l@{\quad}l}
    \Type & ::=^{coind} & & \textbf{Pseudo-type} \\
    &  & \TypeVar & \text{(type variable)}\\
    & | & \nat & \text{(natural type)} \\
    & | & \Type \times \Type & \text{(product)} \\
    & | & \Type \arrow \Type & \text{(arrow)}  \\
    & | & \tbullet\Type & \text{(delay)} 
  \end{array}
 \end{array}
 \]

In addition to the usual constructs of the $\lambda$-calculus,
expressions include constants, ranged over by $\Constant$. Constants are the 
constructor for pairs $\pair$, the projections $\fst$ and  $\snd$
and $\zero$ and $\suc$.
We do not need a primitive constant for the fixed point operator
because it can be expressed and typed inside the language.
Expressions are subject to the usual conventions of the
$\lambda$-calculus. In particular, we assume that the bodies of
abstractions extend as much as possible to the right, that
applications associate to the left, and we use parentheses to
disambiguate the notation when necessary.
We write $\Pair{\Expression_1}{\Expression_2}$ in place of
$\pair~\Expression_1~\Expression_2$. 


 The syntax of \emph{pseudo-types} is
defined  co-inductively. 
A type is
a possibly infinite tree, where each internal node is labelled by a
type constructor $\rightarrow$, 
$\times$ or   $\bullet$  
and has as many children as the arity of the
constructor.
The leaves of the tree (if any) are labelled by basic types which in this case 
are type variables or $\nat$.
The type variables are needed in $\lambdab$ to express the general type of functions such as the identity.
We use a co-inductive syntax to describe infinite data structures (such
as streams).
The syntax for pseudo-types include the types of the simply typed lambda
calculus ({\em arrows} and {\em products})
and the  \emph{delay} type
constructor $\tbullet$~\cite{Nakano00:lics}.

For simplicity, we only include $\nat$ and type variables  as basic types.
One could easily add other basic types such as $\bool$ and $\Unit$
together with constants for their values.

%

\begin{definition}[Types]\label{def:types}
 We say that a pseudo-type $\Type$  is 
  \begin{enumerate}
 
  \item\label{def:types2} 
    {\em regular}  if  
    its tree representation 
     has finitely many distinct sub-trees;
     
  \item\label{def:types3} 
  {\em guarded} if 
  every infinite path in its tree
    representation  has infinitely many $\bullet$'s;

\item  a \emph{type} if
it is regular and  guarded.

  \end{enumerate}
\end{definition}

The regularity condition  implies that we only consider types
admitting a finite representation.
It is equivalent to representing types 
with  $\mu$-notation and a  strong equality
which allows for an infinite number of unfoldings.
This is also called the {\em  equirecursive approach} since it views
types as the unique solutions 
of recursive equations~\cite{DBLP:journals/jfp/GapeyevLP02} \cite[Section 20.2]{Pierce02}.
The existence and
uniqueness of a solution for a pseudo-type satisfying condition \ref{def:types2}  
 follow from known
results~(see \cite{Courcelle83} and also 
Theorem 7.5.34 of \cite{BookCC}). For example, there 
are unique pseudo-types $\badListN$, $\ListN$, and $\infinite$ that respectively
satisfy the equations $\badListN=\nat\times\badListN$,
$\ListN=\nat\times \bullet \ListN$, and $\infinite=\bullet\infinite$.

The guardedness condition intuitively means that not all parts of an
infinite data structure can be available at once: those whose type is
prefixed by a $\bullet$ are necessarily \lq\lq delayed\rq\rq\ in the
sense that recursive calls on them must be deeper.
For example, $\ListN$ is a type 
that denotes streams of natural numbers, where each subsequent element
of the stream is delayed by one $\bullet$ compared to its
predecessor. Instead $\badListN$ is not a type: 
it 
would denote an infinite stream of natural numbers, whose elements are
all available right away.
If the types are written in $\mu$-notation, the guardedness condition means that
all occurrences of $\TypeVar$ in the body $\Type$  of $\mu \TypeVar. \Type$
are in the scope of
a $\bullet$.

The type $\infinite$ is somehow degenerated in that it contains no
actual data constructors. Unsurprisingly, we will see that
non-normalising terms such as
$\omegaterm = (\lambda \var. \var\ \var)(\lambda \var. \var\ \var)$ can
only be typed with $\infinite$
(see Theorem \ref{theorem:weakheadnormalization}).  Without condition~\ref{def:types3},
$\omegaterm$ could be given any type 
since the recursive pseudo-type 
$\fixtype = \fixtype \rightarrow \Type$
would become a type.

 We adopt the usual conventions regarding arrow types (which associate
to the right) and assume the following precedence among
type  constructors:
$\arrow$, $\times$ and  $\bullet$  with
 $\bullet$  having the highest precedence.
 Sometimes we will write $\tbullet[n]\Type$ in place of
  $\smash{\underbrace{\tbullet\cdots\tbullet}_{n\text{-times}}\Type}$.
\\

\subsection{Operational Semantics}
\label{section:reduction}

Expressions reduce according to a standard \emph{call-by-name}
semantics:  
\[
  \begin{array}{c}
    \inferrule[\defrule{r-beta}]{}{
      (\Fun\var\ExpressionE_1)~\Expression_2
      \red
      \Expression_1 \subst{\Expression_2}\var
    }
\quad
  \inferrule[\defrule{r-first}]{}{
      \fst~{\Pair{\Expression_1}{\Expression_2}}
      \red
      \Expression_1
    }    \quad 
    \inferrule[\defrule{r-second}]{}{
      \snd~{\Pair{\Expression_1}{\Expression_2}}
      \red
      \Expression_2
    }  
    \quad
      \inferrule[\defrule{r-ctxt}]{
      \ExpressionE \red \ExpressionF
    }{
      \Context[\ExpressionE]
      \red
      \Context[\ExpressionF]
    } 
\end{array}
\]
where the \emph{evaluation contexts} are 
  $\Context ::=
  \Hole
  \mid ({\Context}~{\Expression})
  \mid \Fst{\Context} \mid \Snd{\Context} \mid \Suc{\Context}
  $. 
{\em Normal forms} are defined as usual
as  expressions that do not reduce.
The reflexive and transitive closure of $\red$ is denoted by $\red^{*}$.

Note that 
$\red$ does not allow  to reduce expressions which are inside the body of an abstraction or 
  component of a pair as  the usual $\beta$-reduction of lambda calculus does.
In the standard terminology of lambda calculus,
 $\red$ actually corresponds to {\em weak head reduction} and 
the  normal forms of $\red$ are actually  called {\em weak head normal forms}.

\subsection{Type System}
\label{section:types}
First we assume a set $\ConsCtx$ containg the types for the constants:
\[
\begin{array}{l@{\,\,\,}l}
\begin{array}[t]{l@{\!\,\,\,}l@{\!\,\,\,}l}
  \pair & : & \Type\arrow\TypeS\arrow\Type\times\TypeS  \\
  \fst & : & \Type\times \TypeS\arrow\Type \\
   \snd & : & \Type\times \TypeS\arrow\TypeS
\end{array}
&
\begin{array}[t]{l@{\!\,\,\,}l@{\!\,\,\,}l}
\zero & : & \nat \\
\suc & : & \nat \arrow \nat
\end{array}
\end{array}
\]

The {\em type assignment system} $\lambdab$ is  
defined by the following  rules. 

\[
      \begin{array}{@{}c@{}}
        \inferrule[\defrule\axiom]{
        \varX: \Type \in \TypeContext
        }{
        \wte{\TypeContext}{\varX}{\Type}
        }
        \qquad
        \inferrule[\defrule\const]{
         \Constant: \Type \in \ConsCtx
        }{
        \wte\TypeContext\Constant\Type
        }
        \quad 
        \inferrule[\defrule\introbullet]{
        \wte{\TypeContext}{\Expression}\Type
        }{
        \wte{\TypeContext}{\Expression}{\tbullet\Type}
        }
\quad
        \inferrule[\defrule\introarrow]{
        \wte{\TypeContext,\varX: \tbullet[n]\Type }{\Expression}{\tbullet[n]\TypeS}
        }{
        \wte{\TypeContext}{\Fun \varX \Expression}{\tbullet[n](\Type \arrow \TypeS)}
        }
        \quad
        \inferrule[\defrule\elimarrow]{
          \wte{\TypeContext}{\Expression_1}{\tbullet[n](\Type \arrow \TypeS)}
          \ \ \ 
          \wte{\TypeContext}{\Expression_2}{\tbullet[n]\Type}
        }{
        \wte{\TypeContext}{\Expression_1\Expression_2}{\tbullet[n]\TypeS}
        }
      \end{array}
  \]

\vspace{0.5cm}

 The  rule~\refrule{\introbullet} 
 introduces   the modality. 
 The rule is  unusual in the sense
 that the expression remains
 the same. We do not have a constructor for  $\tbullet$  at the level
 of expressions.

The rule ~\refrule{\introbullet} does not have any restrictions, i.e. 
 if
it is known that a value will only be available with delay $n$, then
it will also be available with any delay $m \geq n$, but not earlier.
Using rule $\refrule\introbullet$ and the recursive type 
$\fixtype= \tbullet \fixtype \rightarrow t$,
 we can derive that the fixed point
combinator\label{fix}
$
\fix = \lambda \varY. (\lambda \var. \varY\ (\var\ \var) ) (\lambda
\var. \varY\ (\var\ \var) )
$
has type $(\tbullet \Type \rightarrow \Type)\rightarrow \Type$
by assigning 
the type $\fixtype \rightarrow t$
to the first occurrence of $\lambda \var. \varY\ (\var\ \var) $
and  $\tbullet \fixtype \rightarrow t$
to the second one~\cite{Nakano00:lics}.

The  rules\refrule\introarrow{} and ~\refrule{\elimarrow} 
 allow for an arbitrary delay in front of the
types of the entities involved. Intuitively, the number of
$\tbullet$'s represents the delay at which a value becomes
available. So for example, rule~\refrule\introarrow{} says that a
function which accepts an argument $\var$ of type $\TypeT$ delayed by
$n$ and produces a result of type $\TypeS$ delayed by the same $n$ has
type $\tbullet[n](\TypeT \arrow \TypeS)$, that is a function delayed
by $n$ that maps elements of $\TypeT$ into elements of $\TypeS$.

\subsection{Inductive data types}
 \label{section:naturalnumbers}
 
 The temporal modal operator can only be used for defining co-inductive data types.
 Inductive data types need to be introduced separately into the calculus.
For example, we can introduce primitive recursion on $\nat$
 by means of a function 
 $\natrec: \Type\arrow (\nat \arrow \Type \arrow \Type) \arrow \nat \arrow \Type$
 with the reduction rule:  
 \[
\begin{array}{ll}
      \natrec~{\ExpressionF_1}~\ExpressionF_2~\zero
      \red
      \ExpressionF_1
      \quad
      \natrec~{\ExpressionF_1}~\ExpressionF_2~(\suc~\Expression)
      \red
      \ExpressionF_2~\Expression~(\natrec~{\ExpressionF_1}~\ExpressionF_2~\Expression)
    \end{array}
\]

We can also introduce a type $\Lista{\Type}$ for the set of finite lists as an inductive
data type with $\nil :\Lista{\Type}$ and $\consl: \Type \arrow \Lista{\Type} \arrow \Lista{\Type}$
and $\listrec: \TypeS\arrow (\Type \arrow \Lista{\Type} \arrow \TypeS \arrow \TypeS) \arrow \Lista{\Type} \arrow \TypeS$
 with the reduction rule:
 \[
\begin{array}{ll}
      \listrec~{\ExpressionF_1}~\ExpressionF_2~\nil
      \red
      \ExpressionF_1
      \quad
      \listrec~{\ExpressionF_1}~\ExpressionF_2~(\consl~\Expression_1~\Expression_2)
      \red
      \ExpressionF_2~\Expression_1~\Expression_2~(\listrec~{\ExpressionF_1}~\ExpressionF_2~\Expression_2)
    \end{array}
\]

\subsection{Disjoint union}

We could add the  disjoint union $\Type_1 + \Type_2$ of two types $\Type_1$ and $\Type_2$ 
and  add constants:
\[
\begin{array}{ll}
\inl: \Type_1 \arrow \Type_1 + \Type_2
\\
\inr: \Type_2 \arrow \Type_1 + \Type_2
\\
\case: (\Type_1 +\Type_2) \arrow (\Type_1 \arrow \TypeS) \arrow (\Type_2 \arrow \TypeS) \arrow
\TypeS 
\end{array}
\]
the reduction rules:
\[
\begin{array}{ll}
      \case~{(\inl~\Expression)}~{\ExpressionF_1}{\ExpressionF_2}
      \red
      \ExpressionF_1~\Expression
      \quad
      \case~{(\inr~\Expression)}~{\ExpressionF_1}{\ExpressionF_2}
      \red
      \ExpressionF_2~\Expression
    \end{array}
\]
and the evaluation contexts are extended with $\Context ::= ... \mid
 \case~{\Context}~{\ExpressionF_1}{\ExpressionF_2}$.

 Using the disjoint union, we can define the type of  finite and infinite lists as:
 \[\coList{\Type} = \Unit + \bullet \Type \times \coList{\Type}\]

 \subsection{Type of conatural numbers}
 \label{section:conaturalnumbers}
 
 We can also define the type  of conatural numbers  as $\coNat = \Unit + \bullet \coNat$.
Then $\cozero = \inl~\unit$ has type $\coNat$ and $\cosucc = \inr $ has type $\bullet \coNat
\arrow \coNat$.
The set of conatural numbers include $\cosucc^\infty = \fix~\cosucc$.
The usual definition of addition given by the equations:
\[
 n+0 = n \quad  n+ (m+1) = (n+m) +1
 \]
 can be written as 
\[
\add~\varX~\varY = \case~\varY~(\Fun{\varY_1}{\varX})~(\Fun{\varY_1}{\cosucc~(\add~\varX~\varY_1)}) 
 \]
 and it can be typed in $\lambdab$.
 However, the usual definition of substraction given by the equations 
 \[
 n-0 = n \quad 0-n = 0 \quad (n+1) - (m+1) = n -m 
 \]
which is expressed as:
\[
\minus~\varX~\varY = \case~\varX~(\Fun{\varX_1}{\cozero})~
                                 (\Fun{\varX_1}
                                 {\case~\varY~ (\Fun{\varY_1}{\cozero})
                                                            (\Fun{\varY_1}{\minus~\varX_1~\varY_1)})}
\]
is not typable because 
$(\minus~\cosucc^{\infty}~\cosucc^{\infty})$ is not normalizing.

\subsection{Examples of Typable Expressions in $\lambdab$}
\label{section:examples}

We consider  the function $\skipeven$ of the introduction 
  that deletes the elements at even positions of a stream. 
  \begin{equation}
  \label{equation:skip}
  \skipeven = 
  \fix~\Fun{ f x } {\Pair{\Fst{x}}{f~\Snd{\Snd{x}} }}
  \end{equation}
  
  In order to assign the
   type $\ListN \arrow \ListNtwo$ to $\skipeven$, 
 the variable  $f$ has to be delayed once and
  the first occurrence of $\snd$
has to be delayed twice.
Note also  that when typing the application $f~\Snd{\Snd{x}}$
the rule 
 $\refrule{\elimarrow}$ is used with $n=2$.  

\vspace{0.4cm}

The following two functions have type 
$(\Type \arrow \TypeS) \arrow \List{\Type} \arrow \List{\TypeS}$ 
in $\lambdab$.

\begin{center}
\begin{tabular}{l}
$\map~f~x = \Pair{f~(\fst~x)}{(\map \ f \ (\snd~x)}$
\\
$\maap~f~x = \Pair{f~(\fst~x)}{\Pair{f~(\fst~(\snd~x))}{(\maap \ f \ (\snd~(\snd~x))}}$
\end{tabular}
\end{center}

\vspace{0.4cm}

The following  function has type
$\List{\nat} \arrow \List{\nat} \arrow \List{\nat}$ in $\lambdab$.
\[
\addlist~x~y = \Pair{(\fst~x)+ (\fst~y)}{ \addlist~(\snd~x)~(\snd~y)}
\]

\vspace{0.4cm}

The following two functions have type
$\List{\Type} \arrow \List{\Type} \arrow \List{\Type}$ in $\lambdab$.  \\

\begin{tabular}{l}
$\interleave~x~y =  \Pair{\fst~x}{(\interleave \ y \ (\snd~x)}$
\\
$\merge~x~y = \ifkw~(\fst~x) \leq (\fst~y)~\thenkw~ 
\Pair{ \fst~x}{ \merge~(\snd~x)~y}~\elsekw~\Pair{\fst~y}{ \merge~x~(\snd~y)}$
\end{tabular}

\vspace{0.4cm}

The following six functions have type $\List{\nat}$ in $\lambdab$.

\begin{center}
\begin{tabular}{l}
$\ones = \Pair{1}{\interleave~\ones~\ones}$ 
\\
$\nats = \Pair{0}{\map~\suces~\nats}$
\\
$\fib = \Pair{0}{\addlist~{\fib}~{\Pair{1}{\fib}}}$
\\
$\fib' = \Pair{0}{\Pair{1}{\addlist~{\fib'}~{(\snd~\fib')}}}$
\\
$\naats = \Pair{0}{\maap~\suces~\nats}$
\\
$\hamming = \Pair{1}{\merge~(\map~\Fun{\varX}{2 * \varX}~\hamming)~
                            (\merge~(\map~\Fun{\varX}{3 * \varX}~\hamming)
                                   (\map~\Fun{\varX}{5 * \varX}~\hamming)))}$
\end{tabular}
\end{center}

\vspace{0.5cm}

The functions $\merge$, $\ones$, $\fib$, $\fib'$,
$\naats$ and $\hamming$ cannot be typed by the proof assistant Coq because they do not satisfy
the syntactic guardness condition (all recursive calls should be guarded by constructors)
\cite{GimenezCasteranTutorialCoq}.
The function $\fib'$ cannot be typed with sized types \cite{DBLP:conf/lics/Sacchini13}.

It is possible to type in $\lambdab$ other programs shown typable in  
other papers on the modal operator
such as 
$\mkconstant{toggle}$, $\mkconstant{paperfolds}$ \cite{DBLP:conf/fossacs/BizjakGCMB16}.

\subsection{Untypable programs in $\lambdab$}
\label{section:untypable}
We define the  function $\get$  as follows:
\begin{equation}
\label{equation:get2}
\begin{array}{l}
\get = \natrec~\fst~\Fun{\varX \varY \varZ}{\varY~(\snd~\varZ)}
\end{array}
\end{equation}
The first argument of  $\natrec$ has type $\List{\Type} \arrow \Type$ while the second argument
has type $\nat \arrow (\List{\Type} \arrow \Type) \arrow \List{\Type} \arrow \bullet \Type$.
It is not difficult to see that  
$\get$ is not typable in $\lambdab$ unless
$\Type = \bullet \Type$. 
Note that if we re-define $\get$ on the set of conatural numbers and 
use $\fix$ instead of $\natrec$ then we obtain a function $\get'$ defined by:
\[
\get' = \fix (\lambda g. \lambda xy. \case~x~(\lambda x_1. \fst~y)~(\lambda x_1. g~x_1~(\snd~y)))
\]
which is  not typable in $\lambdab$ either.

We now consider the function $\take$ that takes the first $n$ elements of a stream
and returns a finite list in  $\Lista{\Type}$ defined by:
\[
\begin{array}{ll}
\take~0~\varX & =  \nil \\
 \take~(\suc~n)~\varX & = \consl~(\fst~\varX)~(\take~n~(\snd~\varX)) 
\end{array}
\]
It is expressed using $\natrec$ as follows:
\begin{equation}
\label{equation:take2}
\begin{array}{l}
\take = \natrec~\Fun{\varX}{\nil}~\Fun{\varX \varY \varZ}{(\consl~(\fst~\varZ)~(\varY~(\snd~\varZ)))}
\end{array}
\end{equation}
The first argument of $\natrec$ has type $\List{\Type} \arrow \Lista{\Type}$
and the second argument 
has type
$\nat \arrow (\List{\Type} \arrow \Lista{\Type}) \arrow \List{\Type} \arrow \bullet \Lista{\Type}$.
It is not difficult to see that $\take$ is not typable in $\lambdab$.

Though the  type systems by Clouston et al \cite{CBGB15}
and the one   by Atkey et al \cite{AM13}
have to introduce explicit constants for the introduction and elimination of $\bullet$,
they can type programs such as $\get$ and $\take$ which $\lambdabb$ cannot do.

\section{Properties of Typable Expressions}
\label{section:properties}

This section proves  the two most relevant
properties of typable expressions, which are subject reduction
(reduction of expressions preserves their types) and normalisation.
As informally motivated in \cref{section:calculus}, the type constructor
$\bullet$ controls recursion and guarantees normalisation of any
expression that has a type different from $\infinite$. 
This section also proves that any typable expression 
has a  L\`evy-Longo and B\"ohm tree without $\bot$ if the type is $\infinite$-free and
tail finite respectively.

\subsection{Inversion and Subject Reduction}

\begin{lemma}[Weakening]
\label{lemma:weakening}
 If
    $\wte{\TypeContext}{\Expression}{\Type}$ and $\TypeContext \subseteq \TypeContext'$
    then
    $\wte{\TypeContext'}{\Expression}{\Type}$.  
\end{lemma}
\begin{proof} By an easy induction on the derivation.
  \end{proof}
  
\begin{lemma}[Delay]
\label{lemma:delay}
 If
    $\wte{\TypeContext_1, \TypeContext_2}{\Expression}{\Type}$ 
    then
    $\wte{\TypeContext_1, \bullet \TypeContext_2}{\Expression}{\bullet \Type}$.  
\end{lemma}
\begin{proof} By induction on the derivation. We only show the 
case for the rule   \refrule{\introarrow}.
\[  
   \inferrule[\defrule\introarrow]{
        \wte{\TypeContext_1, \TypeContext_2,\varX: \tbullet[n]\Type_1 }{\Expression}{\tbullet[n]\Type_2}
        }{
        \wte{\TypeContext_1, \TypeContext_2}{\Fun \varX \Expression}{\tbullet[n](\Type_1 \arrow \Type_2)}
        }
  \]
  By induction hypothesis,
  $\wte{\TypeContext_1, \bullet \TypeContext_2,\varX: \tbullet[n+1]\Type_1 }{\Expression}{\tbullet[n+1]\Type_2}$.
  By applying the rule \refrule{\introarrow}, 
  we conclude that 
  $ \wte{\TypeContext_1, \bullet \TypeContext_2}{\Fun \varX \Expression}{\tbullet[n+1](\Type_1 \arrow \Type_2)}
  $.
    \end{proof}

\begin{lemma}[Inversion]\label{lem:inv}\mbox{}
  \begin{enumerate}

  \item\label{lem:inv1} If $\wte{\TypeContext}{\Constant}{\Type}$ then
   $\Type \teq  \tbullet^n \Type'$ and $\Constant: \Type' \in \ConsCtx$.

  \item\label{lem:inv2} If $\wte{\TypeContext}{\var}{\Type}$ then
    $\Type \teq  \tbullet^n \Type'$ and
    $\var: \Type' \in \TypeContext$.

  \item\label{lem:inv3} If
    $\wte{\TypeContext}{\Fun \varX \Expression}{\Type}$ then 
    $\Type  \teq \tbullet^n (\Type_1 \arrow \Type_2)$ and
    $\wte{\TypeContext, \varX:
      \tbullet^n\Type_1}{\Expression}{\tbullet^n\Type_2}$

  \item\label{lem:inv5} If
    $\wte{\TypeContext}{\Expression_1 \Expression_2}{\Type}$ then
    $\Type  \teq \tbullet^n\Type_2$ and
    $\wte{\TypeContext}{\Expression_2}{\tbullet^n\Type_1}$ and 
    $\wte{\TypeContext}{\Expression_1}{\tbullet^n(\Type_1 \arrow
      \Type_2)}$
    
  \end{enumerate}
\end{lemma}
\begin{proof} By case analysis and induction on the derivation.
We only show  \cref{lem:inv3}.
 A  
 derivation of
    $\wte{\TypeContext}{\Fun \varX \Expression}{\Type}$
     ends with an
    application of either \refrule{\introarrow}   or 
    \refrule{\introbullet}.
    The proofs for \refrule{\introarrow} is  immediate. 
    If the last applied rule is \refrule{\introbullet} then
    $\Type =  \tbullet\Type'$ and 
    $\wte{\TypeContext}{ \Fun \varX \Expression}\Type'$.
    By induction,
    $\Type' \teq  \tbullet^n (\Type_1 \arrow \Type_2)$ and
    $\wte{\TypeContext,  \varX:
      \tbullet^n\Type_1}{\Expression}{\tbullet^n\Type_2}$.
    Hence, 
    $\Type = \tbullet \Type'  \teq \tbullet^{n+1} (\Type_1 \arrow \Type_2)$ and
    $\wte{\TypeContext, \varX:
      \tbullet^{n+1}\Type_1}{\Expression}{\tbullet^{n+1}\Type_2}$  
      by \cref{lemma:delay}.
\end{proof}

\begin{lemma}[Substitution]\mbox{}
  \label{lemma:substitution}
  If
    $\wte{\TypeContext, \varX: \TypeS}{\Expression}{\Type}$ and
    $\wte{\TypeContext}{\ExpressionF}{\TypeS}$ then
    $\wte{\TypeContext}{\Expression\subst\ExpressionF\varX}{\Type}$.
\end{lemma}

 \begin{proof} By  induction on the structure of expressions.
We only show the case when $\Expression = \lambda \varY. \Expression'$.
 It follows from \ri{\cref{lem:inv}}{\cref{lem:inv3}} that 
 $\Type  \teq \tbullet^n (\Type_1 \arrow \Type_2)$ and
    $\wte{\TypeContext, \varX: \TypeS, \varY: 
      \tbullet^n\Type_1}{\Expression'}{\tbullet^n\Type_2}$.
      By induction hypothesis,
      $\wte{\TypeContext,  \varY: 
      \tbullet^n\Type_1}{\Expression'\subst\ExpressionF\varX}{\tbullet^n\Type_2}$.
   By applying the rule \refrule{\introarrow},
       $\wte{\TypeContext}{\lambda \varY. \Expression'\subst\ExpressionF\varX}{\tbullet^n(\Type_1 \arrow \Type_2)}$.
  \end{proof}

\begin{lemma}[Subject Reduction]
\label{theorem:subjectreduction}
 If
    $\wte{\TypeContext}{\Expression}{\Type}$ 
    and $\Expression \red \Expression'$ then 
     $\wte{\TypeContext}{\Expression'}{\Type}$.
\end{lemma} 

\begin{proof}
By induction on the definition of $\red$. We only do the case 
$(\Fun\varX\ExpressionE_1)~\Expression_2
      \red    
      \ExpressionE_1 \subst{\Expression_2}\varX$.
    Suppose 
    $
    \wte{\TypeContext}{(\Fun\varX\ExpressionE_1)~\Expression_2}{\Type}
    $.
    By \ri{\cref{lem:inv}}{\cref{lem:inv5}} we have that:
   
\[\Type \teq \tbullet^n\Type_2 \quad \quad
      \wte{\TypeContext}{\Expression_2}{\tbullet^n\Type_1}
      \quad \quad
      \wte{\TypeContext}{(\Fun\varX\Expression_1)}{\tbullet^n(\Type_1 \arrow \Type_2)}
   \]
   It follows from \ri{\cref{lem:inv}}{\cref{lem:inv3}}  that
 $   \wte{\TypeContext, \varX: \tbullet^n\Type_1}{\Expression_1}{\tbullet^n\Type_2}$.
    By applying  Substitution Lemma, 
    we deduce that 
    $\wte{\TypeContext}{\Expression_1\subst{\Expression_2}\varX}{\tbullet^n\Type_2}$.   
\end{proof}

Neither Nakano's  type system nor ours is closed under $\eta$-reduction.
For example,
$\varY: \bullet (\Type \arrow \TypeS) \vdash
\Fun\varX {\app{\varY}{\varX}} : (\Type \arrow \bullet \TypeS)$
 but $\varY: \bullet (\Type \arrow \TypeS) \not \vdash
\varY : (\Type \arrow \bullet \TypeS)$.
The lack of subject reduction for $\eta$-reduction  is natural  
  in the context of lazy evaluation where programs are closed terms and only  
  weak head normalised.

\subsection{Normalization}
\label{section:normalization}

We prove that any   expression 
which has a type $\Type$ such that  $\Type \not = \infinite$ 
reduces to a normal form  (Theorem \ref{theorem:weakheadnormalization}).
For this, we define  a type interpretation indexed on the
set of natural numbers for dealing with the temporal operator
$\bullet$. 
The time is discrete and   represented using the set
of natural numbers.
The semantics reflects the fact that one 
$\bullet$ corresponds to one unit of time
 by shifting
the interpretation from $i$ to $i+1$.

Before introducing 
the type interpretation, we give a few definitions.  
Let  $\EE$ be the set of expressions.
We define the following subsets of $\EE$:
\[
\begin{array}{r@{~}l}
 \setWN & =
\{ \Expression \mid \Expression \red^{*} \ExpressionF\ \&  \ 
\ExpressionF \mbox{ is a normal form} \}\\
\setNVAR & = 
\{ \Expression \mid \Expression \red^{*} \Context[\varX] \ \&  \ \varX \mbox{ is a variable} \}
\end{array}
\]

We will do induction on the rank of types. 
 For $\nat$ and type variables,  the rank is always 0.
For the other types, the rank measures the depth  of all what we can observe
at time $0$. We could also compute it by taking 
the maximal $0$-length of all the paths in the
tree representation of the type, where the 
$0$-length of a path is   the number
of type constructors different from $\bullet$
from the root to a leaf or $\bullet$.

\begin{definition}[Rank of a Type]
\label{definition:typerank}
The rank  
of a type $\Type$ (notation $\typerank (\Type)$) is defined as
follows.
\[
\begin{array}{r@{~}ll}
\typerank (\nat) & = 
\typerank(\TypeVar) =  
\typerank (\bullet \Type)  = 0 \\
\typerank (\Type \times \TypeS) & = max(\typerank ( \Type), \typerank(\TypeS)) +1 \\
\typerank (\Type \rightarrow \TypeS) & = max(\typerank ( \Type), \typerank(\TypeS)) +1 
\end{array}
\]
\end{definition}

The rank is well defined (and finite) because 
 the tree representation of
a type cannot have an infinite branch with 
no $\bullet$'s at all
(Condition~\ref{def:types3}
 in Definition \ref{def:types})
 and $\typerank (\bullet \Type)$ is set to $ 0$.

We now define the type interpretation
$\ti \Type  \in \natset \rightarrow \mathcal{P} (\EE)$, which is
an indexed set, where $\natset$ is the set of natural numbers and $ \mathcal{P}$ is the powerset constructor.

\begin{definition}[Type Interpretation]
\label{definition:firsttypeinterpretation}
We define 
$ \indti \Type \ind \subseteq \EE$ by induction on
 $(\ind, \typerank(\Type))$.
\[
\begin{array}{rcl}
\indti\nat \ind & = & \WNVAR \cup  
\{ \Expression \mid \Expression \red^* n \} \\[3pt]
\indti\TypeVar \ind & = & \WNVAR  \\[3pt]
 \indti{ \Type \times \TypeS} \ind & = &  
 \WNVAR \cup    \set{\Expression \mid 
       \Expression \red^* \Pair{\Expression_1}{\Expression_2}   \mbox{ and }
      \Expression_1\in\indti\Type \ind \mbox{ and } \ ~\Expression_2\in\indti\TypeS \ind}\\[3pt]
\indti{\Type \arrow \TypeS}{\ind} & =  &   \WNVAR \cup 
 \set{\Expression \mid 
 \Expression \red^{*} \lambda x. 
 \ExpressionF \mbox{ and }
 \Expression\Expression' \in \indti{\TypeS}{\indj}  \ \ 
 \forall \Expression' \in \indti{\Type}{\indj}, \indj \leq \ind}\\[3pt]
 &  & {} \cup \set{\Expression \mid 
 \Expression \red^{*} \Context[\Constant] \mbox{ and }
 \Expression\Expression' \in \indti{\TypeS}{\indj}  \ \ 
 \forall \Expression' \in \indti{\Type}{\indj}, \indj \leq \ind}
  \\[3pt]
     \indti{\tbullet\Type}{0} & = &     \EE  \\[3pt]
\indti{\tbullet\Type}{\ind+1}  & = & 
           \indti{\Type}{\ind}   
\end{array}
\]
\end{definition} 
\noindent
Note that $\indti{\infinite}{\ind}  =\EE $ for all $\ind \in \natset$.
In the interpretation of the arrow type,  
 the requirement  ``for all  $j \leq i$'' (and not just 
 ``for all $i$'') is crucial
for dealing with the contra-variance of the arrow type in the proof of \ri{\cref{lem:A}}{\cref{lemma:monotonicityofinterpretation}}.

 
\begin{lemma} \mbox{} 
\label{lemma:interpretationofmanybullets}
\begin{enumerate}
\item\label{lemma:interpretationofmanybullets1}  $\indti{\tbullet[n] \Type}{\ind} = \EE$ if $\ind < n$.
\item\label{lemma:interpretationofmanybullets2}  $\indti{\tbullet[n] \Type}{\ind}= \indti{\Type}{\ind -n}$ if $\ind \geq  n$.
\end{enumerate}
\end{lemma}
\begin{proof}
Both items are proved by induction on $n$.
\end{proof}

\begin{lemma}\label{lem:A}
 For all types $\Type$ and $\ind \in \natset$, 
\begin{enumerate}
\item\label{lemma:wnvar}
$\WNVAR \subseteq \indti{\Type}{\ind}$.

\item 
\label{lemma:monotonicityofinterpretation}
$\indti{\Type}{\ind+1} \subseteq \indti{\Type}{\ind}$.

\end{enumerate}

\end{lemma}

\begin{proof} (\cref{lemma:wnvar}) follows by
 induction on $\ind$ and doing case analysis on 
 the shape of the type. 
(\cref{lemma:monotonicityofinterpretation}) follows by induction on
 $(\ind, \typerank(\Type))$. 
 Suppose $\Expression \in  \indti{\Type \arrow  \TypeS}{i+1}$.
 Then, 
$\Expression \Expression' \in \indti{\TypeS}{j}$ 
 for $j \leq i+1$.
 This is equivalent to saying that
 $\Expression \Expression' \in \indti{\TypeS}{j'+1}$
 for $j' \leq i$.
 By induction hypothesis\comma
$\indti{\TypeS}{j'+1} \subseteq \indti{\TypeS}{j'}$.
Hence, 
 $\Expression \in  \indti{\Type \arrow  \TypeS}{i}$.
 \end{proof}

\begin{lemma}
 \label{lemma:interpretationoftypes}
 If $\Type \not =  \infinite$ 
 then
 $\bigcap_{\ind \in \natset} \indti{\Type}{\ind} \subseteq \setWN$.
\end{lemma}

\begin{proof}  
Suppose $\Type = \bullet^n \Type_0$ and $\Type_0$
is either $\nat$, $\TypeVar$, $\Type_1 \arrow \Type_2$ or
$\Type_1 \times \Type_2$. Then, for all $i \geq n$, 
\[
\indti{\bullet^n \Type_0}{i} = \indti{\Type_0}{i-n}  \subseteq \HH
\]
by Lemma \ref{lemma:interpretationofmanybullets}  and the  definition of type interpretation.
 \end{proof}

\bigskip

In order to deal with open expressions we resort
to  substitution functions, as usual.
A substitution function is
a mapping from (a finite set of) variables to $\EE$. We use $\funsubst$ to range over substitution functions.
Substitution functions allows us to  
 extend the semantics to typing 
judgements 
(notation $\TypeContext \modelsi \Expression: \Type$). 

\begin{definition}[Typing  
Judgement Interpretation]
\label{definition:modelsi}
Let $\funsubst$ be a substitution function. 

\begin{enumerate} 

\item $\funsubst \modelsi \TypeContext$
if $\funsubst(\varX) \in 
  \indti{\Type}{\ind}$
 for all $\varX : \Type \in \TypeContext$.
 
\item $\TypeContext \modelsi \Expression: \Type$
if $\funsubst (\Expression) \in  \indti{\Type}{\ind}$
for all $\funsubst \modelsi \TypeContext$.
\end{enumerate}
\end{definition}

As expected we can show the soundness of our type system 
with respect to 
the indexed semantics. 

\begin{theorem}[Soundness]
\label{theorem:soundness} 
If $\wte{\TypeContext}\Expression\Type$
then $\TypeContext \modelsi \Expression:\Type$ for all $\ind \in \mathbb{N}$.
\end{theorem}

The proof of the above theorem can be found in~\cref{appendix:normalisation}.

\begin{theorem}[Normalisation of Typable Expressions] 
\label{theorem:weakheadnormalization}
If $\wte\TypeContext\Expression\Type$ and $\Type\not=\infinite$ then
$\Expression$ reduces (in zero or more steps) to a (weak head) normal form.
\end{theorem}
\begin{proof}
It follows from \cref{theorem:soundness} that
\begin{equation}
\label{equation:modelsi}
\TypeContext \modelsi \Expression : \Type
\end{equation}
for all $\ind \in \mathbb{N}$. 
Let $id$ be the identity substitution and
 suppose $\varX: \TypeS \in \TypeContext$. Then
\[
\begin{array}{lll}
id(\varX) = \varX & \in \WNVAR \\
      &  \subseteq \indti{\TypeS}{\ind} & \mbox{by \ri{\cref{lem:A}}{\cref{lemma:wnvar}}.}
      \end{array}
      \]
This means that $id \modelsi \TypeContext$
for all $\ind \in \mathbb{N}$.
From \cref{equation:modelsi} we have that
$id(\Expression) = \Expression \in \indti{\Type}{\ind} $ for all $\ind$.
Hence,
\[
\Expression \in \bigcap_{\ind \in \natset} \indti{\Type}{\ind}\]
It follows from \cref{lemma:interpretationoftypes} that
 $\Expression  \in\HH$.
\end{proof}

Notice that there are normalising expressions that cannot be typed, for example $\lambda x. \omegaterm \idfun$, where $\omegaterm = (\lambda \varY. \varY\ \varY)(\lambda \varY. \varY\ \varY)$ and $\idfun=\lambda z.z$. 
It is easy to show that $\fix~\idfun$ has type $\infinite$ and so does 
$\omegaterm$  by  \cref{theorem:subjectreduction}.
By Theorem \ref{theorem:weakheadnormalization}, it cannot have other types, and this implies that the application $\omegaterm \idfun$ has no type.

Notice also that there are  normalizing expressions of type $\infinite$, e.g. $x: \infinite \vdash x: \infinite$. However, there are no normalizaing closed expressions of type $\infinite$ as 
the next lemma shows.

\begin{lemma}
If $\vdash \Expression: \Type$ and $\Expression $ is normalizing then 
$\Type \not = \infinite$.
\end{lemma}

\begin{proof}
 By Theorems  \ref{theorem:weakheadnormalization} and \ref{theorem:subjectreduction}, we can assume that $\Expression$ is a closed expression in normal form.
Then, 
 $\Expression$ could be either $\lambda x. \ExpressionF$ or  
  $\Pair{\Expression_1}{\Expression_2}$ or  $\fst$ or  $\snd$ or  $\suc$ or 
  $\suc^{n} \zero$. The type of all these expressions is different from $\infinite$.
\end{proof}

The function $\filter$ that selects the elements of a list $xs$ that satisfy $p$ defined as 
\[
\filter~p~xs = \ifkw~p~(\fst~xs)~\thenkw~\Pair{(\fst~xs)}{(\filter~xs)} \ 
\elsekw \ (\filter~xs)
\]
is normalizing (actually it is in normal form) and it 
can not  be typed in $\lambdab$.
Suppose towards a contradiction that $\filter$ is typable. Then,   
 $(\filter~\ones~(\Fun{x}{\eqels~{x}~{0}}))$  is also typable.
 But $(\filter~\ones~(\Fun{x}{\eqels~{x}~{0}}))$  is not normalizing contradicting
 \cref{theorem:weakheadnormalization}.

\subsection{L\'evy-Longo and B\"ohm Trees}
A simple way to  give meaning to a computation of lambda calculus is to consider 
L\'evy-Longo and B\"ohm trees  \cite{barendregt2012the,DBLP:conf/lambda/Levy75,Longo83}.
 Consider, for example, a procedure for computing the decimal expansions of $\pi$;
  if implemented appropriately, it can provide partial output as it runs and this ongoing output
  is a natural way to assign meaning to the computation. This is in contrast to a program
  that loops infinitely without ever providing an output. These two procedures have very different
  intuitive meanings.
  This section gives a nice characterization  of computations  that never produce meaningless information by means of  types.

  In lambda calculus, expressions such as $\fix~x$ or $\fix~(\lambda x y. x)$
  have no finite normal form though they are intuitively meaningfull
  and should be distinguished from meaningless expressions such as $fix~\idfun$.
By analogy with $\pi$, 
the L\'evy-Longo (as well as the B\"ohm) tree  of a term is  obtained as the limit of these partial 
outputs. 
If in the process of computing the L\'evy-Longo or B\"ohm tree of a term, we find a subexpression
that has no meaning such as $\fix~\idfun$ then this is recorded by replacing 
 $\fix~\idfun$ by $\bot$.
   L\'evy-Longo and B\"ohm trees differ on the notion of 
  meaningless expressions: weak head normal forms for the former and
   head normal forms for the latter.

%
%
%

\begin{definition}[L\`evy-Longo Tree]
\label{definition:llt}
Let $\Expression$ be an expression (it may be untypable).
The L\`evy-Longo tree of $\Expression$, denoted as $\LLT \Expression$ is
defined coinductively as follows.
\[
\begin{array}{ll}
\LLT{\Expression} & = 
\left \{ 
\begin{array}{ll}
x~\LLT{\Expression_1} \ldots \LLT{\Expression_n}  &   \mbox{if }
 \Expression \red^{*} x~\Expression_1 \ldots \Expression_n
\\
\Constant~\LLT{\Expression_1} \ldots \LLT{\Expression_n}  &   \mbox{if }
 \Expression \red^{*} \Constant~\Expression_1 \ldots \Expression_n
\\
\lambda x. \LLT{\Expression'} & \mbox{if }
 \Expression \red^{*}  \lambda x. \Expression' \\
 \bot & \mbox{ otherwise, i.e. } \Expression \mbox{ has no (weak head) normal form }
\end{array}
\right .
\end{array}
\]
\end{definition}

 \begin{definition}[B\"ohm Tree]
 \label{definition:bt}
Let $\Expression$ be an expression (it may be untypable).
The B\"ohm tree of $\Expression$, denoted as $\BT \Expression$ is
defined coinductively as follows.
\[
\begin{array}{ll}
\BT{\Expression} & = 
\left \{ 
\begin{array}{ll}
\lambda x_1 \ldots x_k. 
x~\BT{\ExpressionF_1} \ldots \BT{\ExpressionF_n } & \mbox{ if }
\Hnf{\Expression} = \lambda x_1 \ldots x_k. 
x~\ExpressionF_1 \ldots \ExpressionF_n 
\\
\lambda x_1 \ldots x_k. 
\Constant~\BT{\ExpressionF_1} \ldots \BT{\ExpressionF_n} & \mbox{ if }
\Hnf{\Expression} = \lambda x_1 \ldots x_k. 
\Constant~\ExpressionF_1 \ldots \ExpressionF_n 
\\
 \bot & \mbox{ otherwise }
\end{array}
\right .
\end{array}
\]
where $\Hnf \Expression$ is the  {\em head normal form} of $\Expression$ 
 defined as follows:
\[
\begin{array}{ll}
\Hnf \Expression = & 
\left \{ \begin{array}{ll}
\lambda x_1 \ldots x_k. 
x~\ExpressionF_1 \ldots \ExpressionF_n & \mbox{ if }
\Expression = \Expression_0 \mbox{ and }
\Expression_k \red^{*}  x~\ExpressionF_1 \ldots \ExpressionF_n 
 \\
& \mbox{ and }
\Expression_{i} \red^{*} \lambda x_{i+1}. \Expression_{i+1} \mbox{ for } 0 \leq i \leq k-1
\\
\lambda x_1 \ldots x_k. 
\Constant~\ExpressionF_1 \ldots \ExpressionF_n & \mbox{ if }
\Expression = \Expression_0 \mbox{ and }
\Expression_k \red^{*}  \Constant~\ExpressionF_1 \ldots \ExpressionF_n 
 \\
& \mbox{ and }
\Expression_{i} \red^{*} \lambda x_{i+1}. \Expression_{i+1} \mbox{ for } 0 \leq i \leq k-1
\\
\bot & {otherwise}
\end{array} \right .
\end{array}
\]  

\end{definition}

L\'evy-Longo and B\"ohm trees  are obtained as possible infinite normal forms
of possible infinite $\beta$-reduction sequences.
Infinitary normalization formalizes the idea of productivity, i.e. an infinite reduction
sequence which always produces part of the result \cite{KennawayKSV95,DBLP:journals/tcs/KennawayKSV97,severidevriesICFP2012,DBLP:journals/corr/KurzPSV13}.
Besides the $\beta$-rule, 
the infinitary lambda calculus of L\'evy-Longo trees  has a reduction rule defined as
$M \red \bot$ if $M$ has no weak head normal form while the
one of B\"ohm trees has a similar rule using the condition that $M$ has no head normal form.


\begin{definition}
We say that $\Type$ is {\em $\infinite$-free} if it does not contain $\infinite$.
\end{definition}

For example, $\List{\nat}$ and $\Etype_{\nat}$ are 
 $\infinite$-free  where 
$\Etype_{\Type} = \Type \arrow \bullet \Etype_{\Type}$.
 But $\List{\infinite}$ and $\nat \arrow \infinite$ are not $\infinite$-free.

\begin{theorem}[Productivity I: L\`evy-Longo Trees without $\bot$]
\label{theorem:llt}
Let $\Type$ and all types in $\TypeContext$ be $\infinite$-free.
If $\wte\TypeContext\Expression\Type$ then
the L\`evy-Longo tree of 
$\Expression$ has no $\bot$'s.
\end{theorem}

\begin{proof}
We construct the L\`evy-Longo tree depth by depth.
By \cref{theorem:weakheadnormalization}, $\Expression$ reduces to a 
(weak head) normal form $\Expression_0$ which is typable by \cref{theorem:subjectreduction}.
Suppose $\Expression_0$ is $\lambda x. \Expression'$
(the cases when $\Expression_0 = x~\Expression_1 \ldots \Expression_n$
or $\Expression_0 =\Constant~\Expression_1 \ldots \Expression_n$ are similar).
The L\`evy-Longo tree of $\Expression$ contains  $\lambda x$ at depth $0$.
It only remains to show that we can construct the L\`evy-Longo tree of $\Expression'$ which will be at depth $n >0$.
It follows from  Inversion Lemma that  $\Gamma, x: \Type_1 \vdash \Expression' : \Type_2$
and both $\Type_1 $ and $\Type_2$ are $\infinite$-free.
Hence, we can repeat the process  to obtain the L\`evy-Longo tree of $\Expression'$.
\end{proof}

\begin{definition}\mbox{  }
\begin{enumerate}
\item 
An {\em infinite alternation of $\bullet$'s and $\arrow$'s} is  a type of the form:
$
\Type = \bullet^{n_1} 
(\Type_1 \arrow 
\bullet^{n_2} 
(\Type_2 \arrow \ldots ))
$. 
\item 
We say that $\Type$ is {\em tail finite}
if it is $\infinite$-free and it does not contain 
an infinite alternation  of $\bullet$'s and  $\arrow$'s.
\end{enumerate}
\end{definition}

For example, 
$\dthree = \nat \arrow \bullet (\nat \times \dthree)$ is tail finite but
$\dtwo = \nat \arrow \bullet \dtwo$ is not.

\begin{theorem}[Productivity II: B\"ohm Trees without $\bot$]
\label{theorem:bt}
Let all the types in $\TypeContext$ and $\Type$ be tail finite.
If $\wte\TypeContext\Expression\Type$ then
the B\"ohm tree of 
$\Expression$ has no $\bot$'s.
\end{theorem}

\begin{proof}
By Theorem \ref{theorem:llt}, the L\'evy-Longo tree of $\Expression$ has no $\bot$'s.
Suppose the  L\'evy-Longo tree of $\Expression$ contains a subtree of the form
$\lambda x_1 \lambda x_2 \lambda x_3. \ldots$. 
It is not difficult to prove using \cref{lem:inv} and  \cref{theorem:subjectreduction}
that the type should contain an infinite alternation of $\bullet$'s and $\arrow$'s.
\end{proof}

Let 
$\Expression = \fix~(\lambda x y. x)$. Then $\Expression$ has type $\Etype_{\nat}$
which is $\infinite$-free 
and   $\LLT{\Expression} = \lambda x. \lambda x. \ldots$
 has no $\bot$.
For example, $\fib$ has type $\List{\nat}$ which is tail finite and
$\BT{\fib} = \Pair{0}{\Pair{1}{\Pair{1}{\Pair{2}{\ldots}}}}$
has no $\bot$'s.

\section{Formal Comparision with Other Type Systems}
\label{section:comparision}

This section  stablishes a formal relation between $\lambdab$ and the typed lambda calculi
of   Krishnaswami and Benton \cite{KrishnaswamiB11} and Nakano \cite{Nakano00:lics}.

\subsection{Embedding the Type System by Krishnaswami and Benton into $\lambdab$}
\label{section:kb}
The typing rules of
$\lambdaKB$ defined by Krishnaswami and Benton  \cite{KrishnaswamiB11} are:
\[
\begin{array}{@{}c@{}}
\inferrule[\defrule\axiom]{
  \strut
}{
  \wtei{\TypeContext, \var:_i \Type}{\var}{\Type}{j}
}~~ j \geq i
\qquad
        \inferrule[\defrule\const]{
         \Constant: \Type \in \ConsCtx
        }{
        \wtei{\TypeContext}{\Constant}{\Type}{i}
        }
        \quad 
\inferrule[\defrule\introbullet]{
  \wtei{\TypeContext}{\Expression}{\Type}{i+1}}
  {\wtei{\TypeContext}{\Ebullet \Expression}{\tbullet\Type}{i}}
\qquad
\inferrule[\defrule{$\bullet$E}]{
  \wtei{\TypeContext}{\Expression}{\tbullet \Type}{i}
}{
  \wtei{\TypeContext}{\Await\Expression}{\Type}{i+1}
}
\\ \\
\inferrule[\defrule\introarrow]{
  \wtei{\TypeContext,\varX:_i \Type }{\Expression}{\TypeS}{i}
}{
  \wtei{\TypeContext}{\Fun \varX \Expression}{(\Type \arrow \TypeS)}{i}
}
\qquad
\inferrule[\defrule{\elimarrow}]{
  \wtei{\TypeContext}{\Expression_1}{(\Type \arrow \TypeS)}{i}
  \\
  \wtei{\TypeContext}{\Expression_2}{\Type}{i}
}{
  \wtei{\TypeContext}{\Expression_1\Expression_2}{\TypeS}{i}
}
\end{array}
\]

 $\lambdaKB$ has only the recursive types  $\List{\Type} = \Type \times \bullet \List{\Type}$
 for infinite lists with only one $\bullet$. 
 Since  $\fix$ cannot be expressed
in $\lambdaKB$, we need to add it to the set $\ConsCtx$
of constants:
\[
\fix : (\tbullet\Type\arrow\Type)\arrow\Type
\]

The mapping $\B$ replaces the subindex $i$ by $\bullet^i$ in  typing contexts.
\[
\begin{array}{ll}
\B(\emptyset) & = \emptyset \\
\B(\TypeContext, \varX :_i \Type ) & =  \B(\TypeContext), \varX: \bullet^i t
\end{array}
\]

The function  $\transfKB$ removes the constructor  and destructor of
      $\bullet$ from $\Expression$, i.e. 
      $\TransfKB{\bullet \Expression' } = \Expression'$ and 
      $\TransfKB{\Await{\Expression'}} = \Expression'$.
      
\begin{lemma}[Embedding $\lambdaKB$ into $\lambdab$] 

If $\TypeContext \vdash \Expression :_i \Type$ in $\lambdaKB$
then 
$ \B(\TypeContext) \vdash  \TransfKB{ \Expression }: \bullet^i t $ in~$\lambdab$.


\end{lemma}

\begin{proof} By an easy induction on the derivation. \end{proof}

If only recursive types for lists are available, one cannot create other recursive types such as
trees but having  only one bullet
also limits the amount of functions on streams we can type.
For example, the functions $\skipeven$ (see \cref{equation:skip}), 
$\ones'$ and  $\pairup$ are not typable in $\lambdaKB$ where 
\[
\begin{array}{ll}
\ones' = \Pair{1}{\interleave~\ones'~(\snd~\ones')}
\\
\pairup~xs= \Pair{\Pair{\Fst{xs}}{\Fst{\Snd{xs}}}}{\Snd{\Snd{xs}}}
\end{array}
\]
The functions $\skipeven$, 
$\ones'$ and  $\pairup$ are all typeable in $\lambdab$ because $\lambdab$ has more flexibility
in the location and number of $\bullet$'s which can be inserted in the  types.
 We can derive that $\ones'$ has type  $\List{\TypeVar}$  in $\lambdab$  
 by assigning the type $\List{\TypeVar} \arrow \bullet \List{\TypeVar} \arrow \List{\TypeVar}$ 
 to  $\interleave$  
 and  $\pairup$ has type $\Listhalf{\TypeVar} \arrow \List{\TypeVar}$
 where $\Listhalf{\TypeVar} = \TypeVar \times \TypeVar \times \bullet \Listhalf{\TypeVar}$.

\subsection{Embedding  $\lambdab$ into the Type System by Nakano}
\label{section:nakano}

We consider the type system by Nakano of~\cite{Nakano00:lics} without $\top$
and call it $\lambdaNakano$.

Types in $\lambdaNakano$ are finite, written in  $\mu$-notation and denoted by
$\TypeN$, $\TypeNS$. 
Types  in $\lambdaNakano$ are {\em guarded} which means that
 the type variable $\TypeVar$ of  $\mu \TypeVar. \TypeN$ can occur free in $\TypeN$
 only under the scope of a $\bullet$.

Let $\nakteq$ the equivalence between types which allows for an infinite number of unfolding
and
$\ContextNak$ be a set $\{X_1 \leq Y_1, \ldots, X_n \leq Y_n \}$ of subtype assumptions
between  type variables which   cannot contain twice the same variable. 
The subtyping relation of 
$\lambdaNakano$ is  defined as the smallest relation on types closed under
the following rules:
\[
      \begin{array}{@{}c@{}}
         \inferrule[\defrule{$\leq \bullet$}]{}{
        \ContextNak \vdash  \TypeN \leq \bullet \TypeN
        }
        \quad
         \inferrule[\defrule{$\bullet \arrow$}]{}{
         \ContextNak \vdash \TypeN  \arrow \TypeNS\leq \bullet \TypeN \arrow \bullet \TypeNS
        }
        \quad
        \inferrule[\defrule{$\arrow \bullet$}]{}{
         \ContextNak \vdash 
         \bullet \TypeN \arrow \bullet \TypeNS \leq \tbullet (\TypeN  \arrow \TypeNS)
        }
        \quad 
         \inferrule[\defrule{var}]{}{
       \ContextNak, X \leq Y \vdash   X \leq  Y
        }
        \\ \\
        \inferrule[\defrule{$\nakteq$}]{  \TypeN \nakteq \TypeNS
        }{
       \ContextNak \vdash  \TypeN \leq  \TypeNS
        }   
        \quad
        \inferrule[\defrule{$\mu$}]{
        \ContextNak, X \leq Y  \vdash \TypeN \leq \TypeNS 
        }{
       \ContextNak \vdash  \mu X.  \TypeN \leq \mu Y.  \TypeNS
        }
        \quad
        \inferrule[\defrule{$\bullet$}]{
        \ContextNak  \vdash \TypeN \leq \TypeNS
        }{
       \ContextNak \vdash  \bullet \TypeN \leq \bullet \TypeNS
        }
        \quad
        \inferrule[\defrule{$\arrow$}]{
        \ContextNak_1 \vdash \TypeNS_1 \leq \TypeN_1 \quad \ContextNak_2 \vdash  \TypeN_2 \leq \TypeNS_2
        }{
        \ContextNak_1 \cup \ContextNak_2 \vdash \TypeN_1 \arrow \TypeN_2 \leq  \TypeNS_1 \arrow \TypeNS_2
        }
      \end{array}
\] 
In $\ContextNak \vdash \TypeN \leq \TypeNS $, 
 we assume that all the free type variables of $\TypeN$ and $\TypeNS$ are in $\ContextNak$ 
 and that 
the $X$'s cannot occur on the right hand side of $\leq$ and
the $Y$'s cannot occur on the left.
The typing rules of $\lambdaNakano$ are defined as follows.

\[
      \begin{array}{@{}c@{\qquad}c@{}}
        \inferrule[\defrule\axiom]{
        \varX: \TypeN \in \TypeContext
        }{
        \wte{\TypeContext}{\varX}{\TypeN}
        }
        &
         \inferrule[\defrule{$\leq$}]{
        \wte{\TypeContext}{\Expression}\TypeN
        \ \ \ \TypeN \leq \TypeN'}{
        \wte{\TypeContext}{\Expression}{\TypeN'}
        } 
        \\
\inferrule[\defrule{$\tbullet$ E}]{
         \wte{\bullet \TypeContext}{\Expression}{\bullet \TypeN}
        }{
         \wte{\TypeContext}{\Expression}{\TypeN}
        }
        \inferrule[\defrule\introarrow]{
        \wte{\TypeContext,\varX: \TypeN }{\Expression}{\TypeNS}
        }{
        \wte{\TypeContext}{\Fun \varX \Expression}{(\TypeN \arrow \TypeNS)}
        }
        &
        \inferrule[\defrule\elimarrow]{
          \wte{\TypeContext}{\Expression_1}{\tbullet[n](\TypeN \arrow \TypeNS)}
          \ \ \ 
          \wte{\TypeContext}{\Expression_2}{\tbullet[n]\TypeN}
        }{
        \wte{\TypeContext}{\Expression_1\Expression_2}{\tbullet[n]\TypeNS}
        }
      \end{array}
  \]

We define the function $\munf$ on guarded $\mu$-types as follows.
\[
\begin{array}{ll}
\begin{array}{r@{}l}
\Munf{\TypeVar} & =   \TypeVar \\
\Munf{\bullet \TypeN} & = \bullet \TypeN 
\end{array}
&
\begin{array}{r@{}l}
\Munf{\TypeN \arrow \TypeNS} & =  \Munf{\TypeN} \arrow\Munf{\TypeN \TypeNS} \\
\Munf{\mu \TypeVar. \TypeN} & =  \Munf{\TypeN} \singleTS{\TypeVar}{\TypeN}
\end{array}
\end{array}
\]

We  define the translation $\fromtypeNak$ from $\mu$-types to 
types in $\lambdab$ by coinduction as follows:
\[
\begin{array}{ll}
\toB{\TypeN} =  \left \{
\begin{array}{ll}
\toB{\TypeVar} & \mbox{ if } \Munf{\TypeN}= \TypeVar \\
\bullet \toB{\TypeNS} & \mbox{ if } \Munf{\TypeN}= \bullet \TypeNS\\
 \toB{\TypeN_1} \arrow \toB{\TypeN_2}& \mbox{ if } \Munf{\TypeN}= \TypeN_1 \arrow \TypeN_2 \\
\end{array} \right .
\end{array}
\]
The translation $\toB{\ContextNak}$ is extended to typing contexts in the obvious way.

\begin{lemma}[Embedding $\lambdab$ into  $\lambdaNakano$]
\label{lemma:nakano}
If  $\wte{\toB{\TypeContext}}{\Expression}{\toB{\TypeN}}$ in $\lambdab$
without using rule \refrule{\const} 
then 
$\wte{\TypeContext}{\Expression}{\TypeN}$ in $\lambdaNakano$.
\end{lemma}

\begin{proof} 
It follows by induction on the derivation. It is easy to see that
the typing rules $\refrule{\introarrow}$ and $\refrule{\introbullet}$
of $\lambdab$ are admissible in $\lambdaNakano$ by making use of the 
typing rule \refrule{$\leq$}.
\end{proof}

Programs in $\lambdaNakano$ have more types than in $\lambdab$. For example,
$\vdash \lambda x.x: (\TypeN \arrow \TypeNS) \arrow (\bullet \TypeN \arrow \bullet \TypeNS) $
in $\lambdaNakano$ but
$\not \vdash \lambda x.x: (\TypeN \arrow \TypeNS) \arrow (\bullet \TypeN \arrow \bullet \TypeNS) $
in $\lambdab$.

\section{Denotational Semantics}
\label{section:semantics}

This section gives a denotational  semantics
for $\lambdab$ where types and expressions are interpreted as 
objects and morphisms in the {\em  topos  of trees} \cite{LMCSBirkedaletal}.
We give a self-contained description of this topos 
as a cartesian closed category for a reader familiar with $\lambda$-calculus.

The {\em topos $\topos$ of trees}  has as objects $A$
families of sets $A_1, A_2, \ldots$ indexed by
positive integers, equipped with family of restrictions
$r^{A}_i : A_{i+1} \rightarrow A_i$. 
Types will be interpreted as family of sets (not just sets).
Intuitively the family represents
better and better  sets of approximants for the values of that type.
Arrows $f: A \rightarrow B$
are families of functions $f_i : A_i \rightarrow B_i$
obeying the naturality condition $f_i \circ r_i^{A} =
r_i^{B} \circ f_{i+1}$. 
\[
 \xymatrix@R=1.5em{
  A_1 \ar[d]_{f_1}  & \ar[l]_{r_1^A}  A_2 \ar[d]_{f_2} &  
  \ar[l]_{r_2^A}   A_3   \ar[d]_{f_3}  & \ldots  \\
  B_1 & \ar[l]^{r_1^B}  B_2   & \ar[l]^{r_2^B}  B_3 &  \ldots  }
  \]  
We define $r_{ii}^A = id_A $  and 
 $r_{i j}^{A} = r_{j} \circ \ldots \circ r_{i-1}$
 for $1 \leq j <i$.
Products are defined pointwise.

Exponentials $\expon{A}{B}$  have as components 
the sets: 
\[
(\expon{A}{B})_i   = \{ (f_1, \ldots, f_i) \mid f_j : A_j \rightarrow B_j
\mbox{ and } f_j \circ r_j^{A} =
r_j^{B} \circ f_{j+1}  \}
\]
and as 
restrictions  
$r^{A\Rightarrow B}_i (f_1, \ldots f_{i+1})= (f_1, \ldots, f_i)$.

We define  $\eval : \expon{A}{B} \times A \rightarrow B$ as 
$
\eval_{i} ((f_1, \ldots, f_i), a) = f_i (a)
$ and 
$\curry(f) : C \rightarrow \expon{A}{B}$ 
for $f : C \times A \rightarrow B$
as  
$
\curry (f)_i (c) = (g_1, \ldots, g_i)
$ 
where
$g_j(a) = f_j (r_{ij}^A (c), a)$ for all $a \in A_j$
and $1 \leq j \leq i$.
The functor $\Next : \topos \rightarrow \topos$
is defined on objects as $(\Next A)_1 =\{* \} $ and $(\Next A)_{i+1} = A_i$
where $r_1^{\Next A} = !$ and $r_{i+1}^{\Next A} =r_{i}^{ A}$ 
  and on arrows
  $(\Next f)_1 =id_{\{*\}} $ and $(\Next f)_{i+1} = f_i$.
  We write $\Next^n$ for the $n$-times composition of $\Next$.

  The natural transformation $\next_A : A \rightarrow \Next A $ 
  is given by  $(\next_A)_1 = !$ and $(\next_A)_{i+1} = r_{i}^A$
  which can easily be extended to a natural transformation 
  $\nextn{n}_A : A \rightarrow \Next^{n}  A $ .
It is not difficult to see that
 there are isomorphisms 
 $\isomprod  : \Next A \times  \Next B \rightarrow  \Next(A \times B)$
and 
$ \isomarrow:  \expon{(\Next A)}{(\Next B)} \rightarrow \Next (\expon{A}{B})$
 which are also natural.
 These can also be easily extended to isomorphisms 
$\isomprodn{n}  : \Next^n A \times  \Next^n B \rightarrow  \Next^n (A \times B)$
and 
$ \isomarrown{n}:  \expon{(\Next^n A)}{(\Next^n  B )} \rightarrow \Next^n  ( \expon{A}{B})$.

The category $\setcat$ is a full subcategory of $\topos$ via the functor
$\incFunctor$ with $(\IncFunctor{X})_i = X$ and $r^{\IncFunctor{X}} = id_{X}$
and $(\IncFunctor{f})_i = f$.
The unit $1$ of $\topos$ is $\IncFunctor {*}$.

A type $\Type$ is interpreted as an object in $\topos$
(for simplicity, we omit type variables).

\[
\begin{array}{ll}
\begin{array}{llllll}
\sem{\nat} & = & \IncFunctor{\nat} 
\\
\sem{\Type \times \TypeS} & = & 
\sem{\Type}\times \sem{\TypeS} 
\end{array}
&  \ \ \ \ \ \ 
\begin{array}{llllll}
 \sem{\bullet \Type} & = & \Next \circ \sem{\Type} 
 \\
\sem{\Type \arrow \TypeS}& = & 
\expon{\sem{\Type}}{\sem{\TypeS}}    
\end{array}
\end{array}
\]
In order to justify that 
the interpretation is well-defined,
it is necessary to view the above definition
as indexed sets
and we do 
induction  on  $(i, \typerank(\Type))$
taking the lexicographic order
where 
$\typerank (\Type)$ is defined in \cref{section:normalization}.
By writing the indices
explicitly in the definition of $\sem{\Type}$, we obtain:  
\[\begin{array}{llllll}
\sem{\nat}_i & = &  \nat 
\\
\sem{\Type \times \TypeS}_i  & = & 
\sem{\Type}_i \times \sem{\TypeS}_i \\
\sem{\Type \arrow \TypeS}_i & = & 
\{ (f_1, \ldots, f_i) \mid f_j : \sem{\Type}_j \rightarrow \sem{\TypeS}_j
\mbox{ and } f_j \circ r_j^{\sem{\Type}} =
r_j^{\sem{\TypeS}} \circ f_{j+1}  \}
\\
\sem{\bullet \Type}_{1} & = &
\{ * \} 
\\
 \sem{\bullet \Type}_{i+1}
  & =  & \sem{\Type}_{i}  
 \end{array}
  \]

As it is common in categorical semantics,  the interpretation  is not defined
on lambda terms in isolation but on typing judgements. In order to define terms as morphisms,  we need the context and the type to specify their domain and co-domain.
Typing contexts $\TypeContext =  \varX_1: \Type_1, \ldots \varX_k : \Type_k$
are interpreted as 
$\sem{\Type_1} \times \ldots \times \sem{\Type_k}$.
The interpretation of typed expressions 
$\sem{\TypeContext \vdash \Expression: \Type}: \sem{\TypeContext} \rightarrow \sem{\Type}$
is defined by induction on $\Expression$ (using  Inversion Lemma):

\[
\begin{array}{lll}
\sem{\TypeContext \vdash \varX: \bullet^{n} \Type} &  = &
\nextn{n} \circ \pi_j \mbox{ if $x = x_j$ and $\Type_j = \Type$ and
 $1 \leq j \leq k$}
\\
\sem{\TypeContext \vdash \Fun \varX \Expression:
\bullet^{n} ( \Type_1 \rightarrow \Type_2) } & = &
\isomarrown{n} \circ \curry (\sem{\TypeContext, \varX: \bullet^n \Type_1 \vdash  
\Expression : \bullet^n \Type_2}) 
\\
\sem{\TypeContext \vdash \Expression_1 \Expression_2  : 
\bullet^{n} \TypeS} & = &
\eval \circ \langle 
(\isomarrown{n})^{-1} \circ 
\sem{\TypeContext \vdash \Expression_1 : 
\bullet^n ( \Type \rightarrow \TypeS) }, 
\sem{\TypeContext \vdash \Expression_2 : \bullet^n  \Type} \rangle 
\\
\sem{\TypeContext \vdash \pair:  (\Type \arrow \TypeS \arrow 
 (\Type \times \TypeS)) } & = &
\curry (\curry (!_{\sem{\TypeContext}} \times id_{\sem{\Type \times \TypeS}})) 
\\
\sem{\TypeContext \vdash \fst  : 
 (\Type \times \TypeS \arrow \Type) } & = & \curry(!_{\sem{\TypeContext}} \times \pi_1)
\\
\sem{\TypeContext \vdash \snd  : 
(\Type \times \TypeS \arrow \TypeS) } & = &
 \curry(!_{\sem{\TypeContext}} \times\pi_2)
\\
\sem{\TypeContext \vdash \zero:  \nat } & = & \incFunctor ( * \mapsto 0)
\\
\sem{\TypeContext \vdash \succesor:  (\nat \arrow \nat)} & = & succ
\end{array}
\]

\begin{lemma} [Semantic Substitution]
\label{lemma:semanticsubstitution}
\[
\sem{\TypeContext, \varX: \Type
 \vdash {\Expression_1} : 
 \TypeS }
\circ \langle id, 
\sem{\TypeContext \vdash \Expression_2: \Type}
 \rangle =
\sem {\TypeContext
 \vdash {\Expression_1} \subst{\Expression_2}{\varX}: 
  \TypeS }
\]
\end{lemma}
 
 \begin{proof} This follows by induction on $\Expression_1$.
 We only sketch  the proof for the case $\Expression_1 = \varX$.
\[
\begin{array}{ll}
\sem{\TypeContext, \varX: \Type
 \vdash {\varX} : 
\bullet^{n}  \Type } 
\circ \langle id, 
\sem{\TypeContext \vdash \Expression_2:  \Type}
 \rangle & =
 \nextn{n} \circ \sem{\TypeContext \vdash \Expression_2:   \Type} 
  =  
\sem{\TypeContext \vdash \Expression_2: \bullet^{n}  \Type}
\mbox{}
\end{array}
\tag*{\qedhere}
\]
 \end{proof}

\begin{theorem}[Soundness]
\label{theorem:soundnesssemantics}
If $\wte{\TypeContext}{\Expression}{\Type}$
and $\Expression \red \Expression'$ then
$\sem{\wte{\TypeContext}{\Expression}{\Type}}=
\sem{\wte{\TypeContext}{\Expression'}{\Type}}$.
\end{theorem}

\begin{proof}
We show  the case of $\defrule{r-beta}$.
Let 
$v_1 = \sem{\TypeContext, \varX: \bullet^n \Type_1
 \vdash {\Expression_1} : 
\bullet^n  \Type_2}$
 and 
$v_2  = \sem{\TypeContext \vdash \Expression_2: \bullet^n  \Type_1}$.
\[
\begin{array}[b]{lll}
\sem{\TypeContext \vdash (\Fun{\varX}{\Expression_1})\Expression_2  : 
\bullet^{n} \Type_2} & = &
\eval \circ \langle 
(\isomarrown{n})^{-1} \circ \isomarrown{n}
 \circ \curry (v_1) , v_2 \rangle 
 =  
v_1 \circ \langle id, v_2 \rangle 
\\
&  =  &
\sem {\TypeContext
 \vdash {\Expression_1} \subst{\Expression_2}{\varX}: 
\bullet^n  \Type_2
} 
\mbox{ by Lemma \ref{lemma:semanticsubstitution}}
\end{array}
\tag*{\qedhere}
\]
\end{proof}

%
%
%

We consider contexts 
 $\PContext$  without any restriction in the position of the hole $[]$.
We say that $\PContext$ is a {\em closing context of type $\Type$} if  
$\wte{x:\TypeS}{\PContext[x]}{\Type}$ for some variable $x$ that does not occur in
$\PContext$.

\begin{definition} Let $\wte{\TypeContext}{\Expression_1}{\Type}$
and $\wte{\TypeContext}{\Expression_2}{\Type}$.
We say that $\Expression_1 $ and $ \Expression_2$ 
 are {\em observationally equivalent}, denoted by 
 $\Expression_1 \oeqctx \Expression_2$, if 
 $\PContext[\Expression_1] \red^* \succesor^{n}~\zero$ 
iff  
 $\PContext[\Expression_2] \red^* \succesor^{n}~\zero$ 
 for all closing contexts
$\PContext$ of type $\nat$.
\end{definition}

\begin{corollary}
If $\sem{\wte{\TypeContext}{\Expression_1}{\Type}}= 
\sem{\wte{\TypeContext}{\Expression_2}{\Type}}$
then, 
$\Expression_1 \oeqctx \Expression_2$.
\end{corollary}

\begin{proof}
By compositionality of the denotational semantics,
$\sem{\wte{}{\PContext[\Expression_1]}{\nat}}= 
\sem{\wte{}{\PContext[\Expression_2]}{\nat}}$.
Suppose $\PContext[\Expression_1] \red^* \succesor^{n}~\zero$.
By \cref{theorem:soundnesssemantics}, 
$\sem{\wte{}{\PContext[\Expression_1]}{\nat}}= 
\sem{\succesor^{n}~\zero}$.
By \cref{theorem:weakheadnormalization}, $\PContext[\Expression_2] \red^*\succesor^{m}~\zero$.
It is not difficult to show using the definition of interpretation that $m=n$.
\end{proof}

\section{A Type Inference Algorithm}
\label{section:typeinference}

In this section, we  define a 
 type inference algorithm 
 for $\lambdab$.
Apart from the usual complications 
that come from having no type declarations, 
the difficulty of finding an appropriate type inference algorithm for
$\lambdab$ 
is due to the fact that 
the expressions  do not have a constructor and destructor for $\tbullet$.
We do not know which sub-expressions need to be delayed as
illustrated by the type derivation of $\fix$ where the first
occurrence of 
$(\lambda \var. \varY\ (\var\ \var) )$
has a different derivation  from   the second one
since $\defrule{\introbullet}$ is applied in different places \cite{Nakano00:lics}.
Even worse, 
in case a sub-expression has to be delayed, we do not know  
how many times  needs to  be delayed to be able
to type the whole expression.

\subsection{A Syntax-directed Type System}
\label{sec:syntaxdirected}

We obtain a syntax-directed type system by  eliminating the rule
\refrule{\introbullet}
from $\lambdab$.
The {\em type assignment system} $\lambdabbminus$ is   defined by the following  rules. 
\[
      \begin{array}{@{}c@{}}
        \inferrule[\defrule\axiom]{
        \varX: \Type \in \TypeContext
        }{
        \wte{\TypeContext}{\varX}{\bullet^n \Type}
        }
        \qquad
        \inferrule[\defrule\const]{
         \Constant: \Type \in \ConsCtx
        }{
        \wte\TypeContext\Constant\bullet^n\Type
        }
        \quad \quad 
         \inferrule[\defrule\introarrow]{
        \wte{\TypeContext,\varX: \tbullet[n]\Type }{\Expression}{\tbullet[n]\TypeS}
        }{
        \wte{\TypeContext}{\Fun \varX \Expression}{\tbullet[n](\Type \arrow \TypeS)}
        }
\quad \quad    
        \inferrule[\defrule\elimarrow]{
          \wte{\TypeContext}{\Expression_1}{\tbullet[n](\Type \arrow \TypeS)}
          \ \ \ 
          \wte{\TypeContext}{\Expression_2}{\tbullet[n]\Type}
        }{
        \wte{\TypeContext}{\Expression_1\Expression_2}{\tbullet[n]\TypeS}
        }
      \end{array}
  \]

\begin{lemma}[Equivalence between $\lambdabb$ and $\lambdabbminus$]
\label{lemma:lambdabbminus}

$\wte{\TypeContext}{\Expression}{\Type}$ in $\lambdabbminus$
iff $\wte{\TypeContext}{\Expression}{\Type}$ in $\lambdabb$.
\end{lemma}

\begin{proof} The direction from left to right follows by  induction on the derivation.
The direction from right to left follows  by induction on $\Expression$ using Lemma \ref{lem:inv}.
\end{proof}

Though $\lambdab$ and $\lambdabbminus$
 are equivalent in the sense they type the same expressions, they do not have the same
type derivations. The derivations 
of $\lambdabbminus$ are more restrictive than the ones of 
$\lambdab$ since one can only introduce $\bullet$'s  at the beginning of the derivation with the 
rules \refrule{\axiom} and \refrule{\const}.

\subsection{Meta-types}
\label{sec:metatypes}

The type inference algorithm infers {\em meta-types}
which are  a generalization of types
where the $\bullet$ can be exponentiated 
with  integer expressions, e.g.
$\bullet^{\IntVar_1} \TypeVar \rightarrow \bullet^{\IntVar_2-\IntVar_1 } \TypeVar$.
The syntax for {\em meta-types} is defined below.
A  meta-type can contain type variables and
(non-negative) integer expressions with variables.
In this syntax,   
 $\bullet \MetaType$ is written as  $\bullet^{1} \MetaType$.
We  identify $\bullet^{0} \MetaType$ with $\MetaType$ and
$\tbullet^{\IntExp} \tbullet^{\IntExp'} \MetaType$ with
 $\tbullet^{\IntExp + \IntExp'} \MetaType$.

\[
\begin{array}{@{}l@{\qquad}l@{}}
\begin{array}[t]{r@{~~}c@{~~}l@{\quad}l}
    \IntExp & ::=^{ind} & & \textbf{Integer Expression} \\
    &   & \IntVar & \text{(integer variable)} \\
    &  | & \mkconstant{n} & \text{(integer number)} \\
    & | & \IntExp + \IntExp  & \text{(addition)} \\
    & | & \IntExp - \IntExp & \text{(substraction)}  
  \end{array}
   \quad
  \begin{array}[t]{r@{~~}c@{~~}l@{\quad}l}
    \MetaType & ::=^{coind} & & \textbf{Pseudo Meta-type} \\
    &   & \TypeVar & \text{(type variable)} \\
   &  | & \nat & \text{(natural numbers)} \\
    & | & \MetaType \times \MetaType & \text{(product)} \\
    & | & \MetaType \arrow \MetaType & \text{(arrow)}  \\
    & | & \tbullet^{\IntExp} \MetaType & \text{(delay)}
  \end{array}
\end{array}
\]

\begin{definition}
We say that $\MetaType$ is {\em positive} if 
 all the exponents of $\tbullet$
are natural numbers. 
\end{definition}

We identify a positive pseudo meta-type with the pseudo-type 
obtained from replacing  the type constructor 
$\bullet^{n} $ in the syntax of pseudo meta-types  with $n$ consecutive 
$\bullet$'s  in the syntax of pseudo-types.

\begin{definition}[Integer Expression and Type Substitution]
Let  $\substT = \mapTS{\TypeVar_1}{\MetaType_1}{\TypeVar_n}{\MetaType_n}$ be a finite  mapping from type variables to  pseudo meta-types  where $\dom (\substT) = \{ \TypeVar_1, \ldots, \TypeVar_n \}$.
Let $\substE = \mapTS{\IntVar_1}{\IntExp_1}{\IntVar_m}{\IntExp_m}$ be a finite mapping 
from integer variables to integer expressions where
$\dom(\substE) = \{\IntVar_1, \ldots, \IntVar_m \}$.
Let also $\substTE = \substT \cup \substE$.
We define the  substitutions $\appTS{\substE}{\IntExp}$ on an integer expression $\IntExp$
and $\appTS{\substTE}{\MetaType}$ on a pseudo meta-type $\MetaType$ as follows.
\[
\begin{array}{lll}
\begin{array}{lll}
\appTS{\substE}{\IntVar} & = &
\left \{ 
\begin{array}{ll}
 \IntExp & \mbox{ if $\newTSubs{\IntVar}{\IntExp} \in  \substE$ }\\
\IntVar & \mbox{ otherwise}
\end{array} \right.
\\
\appTS{ \substE}{ \mkconstant{n} } & = & \mkconstant{n} \\
\appTS{ \substE }{\IntExp_1 + \IntExp_2} & = & \appTS{ \substE}{\IntExp_1} + \appTS{\substE}{\IntExp_2} \\
\appTS{ \substE}{ \IntExp_1 - \IntExp_2} & = & \appTS{ \substE}{\IntExp_1} - \appTS{\substE}{\IntExp_2}
\end{array} 
 \quad
 \begin{array}{lll}
\appTS{\substTE} {\TypeVar} & = & 
\left \{ \begin{array}{ll}
\MetaType  & \mbox{ if $\newTSubs{\TypeVar}{\MetaType} \in  \substT$ }\\
 \TypeVar & \mbox{ otherwise}
\end{array} \right.
 \\
 \appTS{\substTE} {\nat} & = &  \nat \\ 
\appTS{\substTE}{\MetaType_1 \arrow \MetaType_2} & = & 
\appTS{\substTE}{\MetaType_1} \arrow \appTS{\substTE}{\MetaType_2} \\
\appTS{\substTE}{\MetaType_1 \times  \MetaType_2} & = & 
\appTS{\substTE}{\MetaType_1} \times  \appTS{\substTE}{\MetaType_2} \\
\appTS{\substTE}{ \tbullet[\IntExp] \MetaType} & = & 
\tbullet[\appTS{\substE}{\IntExp}] \appTS{\substTE}{\MetaType}
\end{array}\end{array}
\]
If the substitution  $\substT$ is $  \singleTS{\TypeVar}{\MetaType}$ then 
$\substT(\MetaType) $ is also denoted as $\MetaType \singleTS{\TypeVar}{\MetaType}$ as usual.
\end{definition}

Composition of substitutions is defined as usual. Note that
the empty set is the identity substitution and 
$\substTE = \substE \circ \substT $ if $\substTE = \substT \cup \substE$.

\begin{definition}
We say that $\substE$ is {\em natural} if 
it maps all integer variables into
natural numbers, i.e. 
$\substE (\IntVar) \in \mathbb{N}$ for all $\IntVar \in \dom(\substT)$.
\end{definition}

For example,
the substitution $\substE = \{\IntVar\mapsto 3, \IntVarM \mapsto 5\}$
is natural  but $\appTS{\substE}{\bullet^{\IntVar - \IntVarM} \nat}$ is not positive.

We can define {\em regular}   pseudo meta-types similar
to  \cref{def:types}. We say that a pseudo meta-type $\MetaType$ is {\em guarded}
if $\substE(\MetaType)$ is positive and guarded (as a pseudo type) 
for all natural $\substE$.

The pseudo meta-type which is solution of the recursive equation 
$\MetaType = \nat \times \bullet^{\IntVar} \MetaType$
is not guarded but it is guarded if we substitute $\IntVar$ by $\IntVarM +1$.

\begin{definition}
A {\em meta-type} is a pseudo meta-type that is regular and guarded.
\end{definition}

Note that finite pseudo meta-types are always meta-types, e.g. $\nat \times \bullet^{\IntVar} \TypeVar$
is a meta-type.

\begin{definition}
We say that the substitution $\substTE = \substT \cup \substE$  is:
 \begin{enumerate}
 
 \item {\em guarded} if 
 $\substTE(\TypeVar)$ is guarded
 for all $\TypeVar \in \dom(\substT)$.

 \item {\em meta-type} if  it is a substitution from 
 type variables  to meta-types.
 
 \item  
 {\em positive} if $\substE$ is natural   and  $\substE (\substT(\TypeVar))$ is positive 
 for all $\TypeVar \in \dom(\substT)$.

  \item {\em type} if it is meta-type and positive.

 \end{enumerate}
\end{definition}

Note that a type substitution $\substTE = \substT \cup \substE$ is a substitution from 
 type variables  to types where $\substE$ is natural.

\begin{definition}[Constraints]
\label{definition:constraints}
 There are two types of {\em constraints}.

\begin{enumerate}
\item A {\em meta-type  constraint} is $\MetaType \equal \MetaType'$
     for finite meta-types $\MetaType$ and $ \MetaType'$.

\item An {\em integer constraint} is either $\IntExp \equal \IntExp'$
or $\IntExp \mayor \IntExp'$ or  $\IntExp \menor \IntExp'$.
\end{enumerate}

  \end{definition}
 
Sets of constraints  are denoted by $\setC, \setE$, etc.
If $\setC$ is a set of meta-type constraints then,

\begin{enumerate}

\item $\eqc{\setC}$  denotes the subset of equality constraints from $\setC$,

 \item $\intc{\setC}$ denotes the set of integer constraints from $\setC$.
 
 \end{enumerate}
 
Moreover,  $\substTE \models \MetaType \equal \MetaType'$
means that
 $\appTS {\substTE}{\MetaType} \treeEq \appTS{\substTE}{\MetaType'}$.
Similarly, we define 
$\substE \models\IntExp \equal \IntExp'$,
 $\substE \models\IntExp \mayor \IntExp'$ and  $
\substE  \models\IntExp \menor \IntExp'$. 
This notation extends to sets $\setC$
 of constraints  in the obvious way.
 When $\substTE \models \setC$, we say that $\substTE$ is a {\em solution} ({\em unifier}) for
 $\setC$. 
 

 \begin{definition}
We say that $\substT$ is:

\begin{enumerate}

\item $\setE$-{\em guarded} if $\substE \circ \substT$ is guarded for all
natural substitution $\substE$ such that $\substE \models \setE$.

\item $\setE$-{\em positive} if $\substE \circ \substT$ is positive for all
natural substitution $\substE$ such that $\substE \models \setE$.

\end{enumerate}
\end{definition}

Recall that  $\stream{\IntVar}{\TypeVarY} =
   \TypeVarY \times \bullet^{\IntVar} \stream{\IntVar}{\TypeVarY}
   $.
For example,

\begin{enumerate}

\item  $\substT = \{ \TypeVar \mapsto \stream{\IntVar}{\TypeVarY} \}$ is $\setE$-guarded 
but it is not $\emptyset$-guarded for
$\setE = \{ \IntVar \mayor 1 \}$ and,

\item 
 $\substT = \{ \TypeVar \mapsto \bullet^{M-N} \nat \}$ is $\setE$-positive 
but it is not $\emptyset$-positive for 
$\setE = \{ M \mayor N \}$.

\end{enumerate}

\begin{definition}

\label{definition:substitutional}
 We say that $\setC$ is \emph{substitutional}
 if 
 $\eqc{\setC }= \{\TypeVar_1 \equal \MetaType_1, \ldots, \TypeVar_n \equal \MetaType_n \}$, 
 all variables $\TypeVar_1, \ldots, \TypeVar_n$
 are pairwise different.
 \end{definition}
 
 Since 
 a substitutional $\setC$ corresponds to a set of recursive equations
 where $\MetaType_1, \ldots, \MetaType_n$ are finite
  meta-types, there  
 exists a unique solution 
  $\substT_{\setC}$ 
 such that $\substT_{\setC} (\TypeVar_i)$  is regular
  for all $1 \leq i \leq n$ 
  \cite{Courcelle83}, 
  \cite[Theorem 7.5.34]{BookCC}. 
  The mapping $\substT_{\setC}$ is the {\em most general unifier}, i.e.
  if $\substTE \models \setC$ then there exists $\substTE'$ such that
  $\substTE = \substTE' \circ \substT_{\setC}$. Moreover, if $\substTE = \substT \cup \substE$
  then $\substTE' = \substT' \cup \substE$ where
  \[
  \begin{array}{ll}
  \substT'(\TypeVar) & = 
  \left \{ \begin{array}{ll} 
            \substT(\TypeVar) & \mbox{ if }
               \TypeVar \neq \TypeVar_i \mbox{ for all } 1 \leq i \leq n \\
              \TypeVar & \mbox{otherwise} 
            \end{array}
            \right .
    \end{array}
            \]
  For example, if $\setC = \{ \TypeVar \equal  \TypeVarY \times \bullet^{\IntVar} \TypeVar \}$
  then $\substT_{\setC} (\TypeVar) = \stream{ \IntVar}{\TypeVarY}$.
  If $\substTE \models \setC$ and $\substTE = \substT \cup \substE$
  then 
   $\substTE = \substTE' \circ \substT_{\setC}$
  and  
$\substTE' = \substT' \cup \substE$   and 
  $\substT' =    \singleTS{\TypeVarY}{\substT(\TypeVarY)}$.
  
  Note that 
  the unicity of the solution of a set of recursive equations 
  would not be guaranteed if we were
  following  the iso-recursive approach 
which  allows only for a finite number of unfoldings 
$\mu \TypeVar. \Type =
\Type \singleTS{\TypeVar}{ \mu \TypeVar. \Type} $ 
\cite{DBLP:journals/iandc/CardoneC91,BookCC}.

For $1 \leq i \leq n$,
we know that 
$\substT_{\setC} (\TypeVar_i)$  is regular 
but we do not know yet if $\substT_{\setC} $ is guarded or  positive.
Below, we show some examples:

\begin{enumerate}

\item Let $\setC = \{ \TypeVar_1  \equal \TypeVar_2 \arrow \bullet \TypeVar_1, 
\TypeVar_2  \equal \bullet \TypeVar_1 \arrow \bullet \TypeVar_2\}$.
Then, $\substT_{\setC}= \{\TypeVar_1 \mapsto \Type_1, \TypeVar_2 \mapsto \Type_2 \}$ is a type substitution where $\Type_1$ and $\Type_2$ satisfy the recursive equations
$\Type_1 =  \Type_2 \arrow \bullet \Type_1$ and $\Type_2 = \bullet \Type_1 \arrow \bullet \Type_2$. 

\item Let  $\setC = \{ \TypeVar_1  \equal \TypeVar_2 \arrow \bullet \TypeVar_1, 
\TypeVar_2  \equal \TypeVar_1 \arrow \bullet \TypeVar_2\}$.
Then,  $\substT_{\setC} = \{\TypeVar_1 \mapsto \Type_1, \TypeVar_2 \mapsto \Type_2 \}$ 
is positive
where $\Type_1$ and $\Type_2$ satisfy the recursive equations
$\Type_1 =  \Type_2 \arrow \bullet \Type_1$ and $\Type_2 =  \Type_1 \arrow \bullet \Type_2$. 
But $\substT_{\setC} $ is not guarded
since there is an infinite branch in the tree representations of $\Type_1$ and
$\Type_2$ that has no  $\bullet$'s. Observe that 
$\Type_1 =  (\Type_1 \arrow \bullet \Type_2) \arrow \bullet \Type_1$ and 
$\Type_2 = (\Type_2 \arrow \bullet \Type_1) \arrow \bullet \Type_2$.

\item Let 
$\setC = \{ \TypeVar_1  \equal \bullet^{\IntVar +1} (\TypeVar_2 \arrow \bullet \TypeVar_1), 
\TypeVar_2  \equal \bullet^{\IntVarM} \TypeVar_1 \}$.
Then, $\substT_{\setC} = \{\TypeVar_1 \mapsto \MetaType_1, \TypeVar_2 \mapsto \MetaType_2 \}$ is a  meta-type substitution where $\MetaType_1 $ and $\MetaType_2$ satisfy the recursive equations
$\MetaType_1  = \bullet^{\IntVar +1} (\MetaType_2 \arrow \bullet \MetaType_1)$ and 
$\MetaType_2 = \bullet^{\IntVarM} \MetaType_1$.
But  $\substT_{\setC} $ is not positive since $\IntVar$ and $\IntVarM$ are integer variables.

\end{enumerate}

\begin{definition} \label{definition:Calwaysnormal}
We say that $\setC$ is 

\begin{enumerate}

\item {\em always positive} if   
all the  meta-types in $\appTS{\substE}{\eqc{\setC}}$ are positive 
  for all  natural  substitution $\substE$ such that $\substE \models \intc{\setC}$.
  
\item {\em simple} if  $0 \menor \IntExp  \in \setC$
for all $\TypeVar \equal \bullet^{\IntExp} \TypeVar' \in \setC$. 

\end{enumerate}

\end{definition}

For example, 
$\setC = \{ \TypeVar \equal \bullet^{\IntVar} (\TypeVar_1 \arrow \TypeVar_2), \IntVar \mayor 1 \}$
is always positive while
$\setC = \{ \TypeVar \equal \bullet^{\IntVar} (\TypeVar_1 \arrow \TypeVar_2)\}$ is not.

\subsection{Constraint Typing Rules}
\label{subsection:constraint}

The {\em constraint typing rules} for $\lambdabb$ are given in  \cref{table:constraints}.
Since $\pair$, $\fst$ and $\snd$ are actually `polymorphic', we need
to generate different fresh variables depending on their position.
For convenience, we define the following function:
\[
\begin{array}{ll}
\typeOf{\pair}{\TypeVar_1}{\TypeVar_2} & = 
\TypeVar_1 \arrow \TypeVar_2 \arrow \TypeVar_1 \times \TypeVar_2
\\
\typeOf{\fst}{\TypeVar_1}{\TypeVar_2} & = 
\TypeVar_1 \times \TypeVar_2\arrow  \TypeVar_1 
\\
\typeOf{\snd}{\TypeVar_1}{\TypeVar_2} & = 
\TypeVar_1 \times \TypeVar_2\arrow  \TypeVar_2
\end{array}
\]

\begin{table}[!t]
\caption{Constraint Typing Rules for $\lambdab$}
\label{table:constraints}
\begin{framed}
$
\begin{array}{@{}l@{}l}
\inferrule[\defrule{\axiom}]{
  \IntVar \mbox{ is fresh} 
}{
  \typec{\TypeContextD, \varX: \TypeVar}{\varX}{\tbullet^\IntVar \TypeVar}{\emptyset}
}  
\qquad 
  \inferrule[\defrule{\prodconst}]
  {\IntVar, \TypeVar_1, \TypeVar_2 \mbox{ are fresh}  }{
          \typec{\TypeContextD}{\Constant}
          {\bullet^{\IntVar} \typeOf{\Constant}{\TypeVar_1}{\TypeVar_2} }{\emptyset}
        }
         \Constant \in \{\pair, \fst, \snd \} 
        \\\\      
        \inferrule[\defrule{\natzero}]
  {\IntVar \mbox{ is fresh}}{ \typec{\TypeContextD}{\zero}
          {\bullet^{\IntVar} \nat }{\emptyset}
        }
        \quad
        \inferrule[\defrule{\natsucc}]
  {\IntVar \mbox{ is fresh}}{ \typec{\TypeContextD}{\suc}
          {\bullet^{\IntVar} (\nat \arrow \nat) }{\emptyset}
        }
\\
\\
\inferrule[\defrule\introarrow]{
  \typec{\TypeContextD,\varX: \TypeVar }{\Expression}{\MetaType}{\setC} 
  \\
  \TypeVar, \TypeVar_1, \TypeVar_2, \IntVar \mbox{ are fresh}
  }
  {
  \typec{\TypeContextD}{\Fun \varX \Expression}
  {\tbullet[\IntVar](\TypeVar_1 \arrow \TypeVar_2)}
  {\setC \cup \{\TypeVar \equal \tbullet[\IntVar] \TypeVar_1 , \MetaType \equal \tbullet[\IntVar] \TypeVar_2\}}
} 
\\\\
\inferrule[\defrule\elimarrow]{
  \typec{\TypeContextD}{\Expression_1}{\MetaType_1}{\setC_1}
  \\
  \typec{\TypeContextD}{\Expression_2}{\MetaType_2}{\setC_2}
  \\
  \TypeVar_1, \TypeVar_2, \IntVar \mbox{ are fresh}
}{
  \typec{\TypeContextD}{\Expression_1\Expression_2}
  {\tbullet[\IntVar]\TypeVar_2}
  {\setC_1 \cup \setC_2  \cup 
  \{ \tbullet [\IntVar](\TypeVar_1 \arrow \TypeVar_2) \equal \MetaType_1,
  \tbullet [\IntVar] \TypeVar_1 \equal \MetaType_2
   \}
  }
}      
\end{array}
$
\end{framed}

\end{table}

The typing rules should be read bottom-up. We start from an empty context and 
a closed expression and we build 
the typing context $\TypeContextD$ and the set $\setC$ of constraints.
We can, then, assume that the typing context  $\TypeContextD$
only contains declarations  of the form $\var:\TypeVar$ and 
the set $\setC$ contains only equality constraints between
finite meta-types.The type variables in the contexts $\TypeContextD$ are all
fresh and they are created by the  rule
\refrule{\introarrow}  for the abstraction.

If instead of generating the fresh variable $\TypeVar$ in
\defrule{\introarrow} 
we  directly put $\tbullet[\IntVar] \TypeVar_1$ in the context,
then the algorithm would not be complete.
For example,
consider
$x: \bullet^{5} \nat \vdash x: \bullet^{7} \nat$.
Then $\IntVar$ should  be assigned the value $3$ and not $5$
 if later we  derive that 
 $\lambda x. x : \bullet^{2} (\bullet^{3} \nat \arrow \bullet^{5} \nat)$.

\begin{theorem}[Soundness of Constraint Typing]
\label{theorem:soundnessconstraint}
Let $\substTE$ be a  type  substitution.
If
 $\typec{\TypeContextD}{\Expression}{\MetaType}{\setC}$ and 
$\substTE \models \setC$ then
$\wte{ \appTS{\substTE}{\TypeContextD}}{\Expression}
{\appTS{\substTE}{\MetaType}}$ in $\lambdabbminus$
and $\appTS{\substTE}{\MetaType}$ is a type.
\end{theorem}

\begin{proof}
We proceed by induction on the derivation of
$\typec{\TypeContextD}{\Expression}{\MetaType}{\setC}$.
Suppose the last rule in the derivation of
$\typec{\TypeContextD}{\Expression}{\MetaType}{\setC}$ is:
\[
\inferrule[\defrule\introarrow]{
  \typec{\TypeContextD,\varX: \TypeVar }{\Expression}{\MetaType'}{\setC} 
  \\
  \TypeVar, \TypeVar_1, \TypeVar_2, \IntVar \mbox{ are fresh}
  }
  {
  \typec{\TypeContextD}{\Fun \varX \Expression}
  {\tbullet[\IntVar](\TypeVar_1 \arrow \TypeVar_2)}
  {\setC \cup \{\TypeVar \equal \tbullet[\IntVar] \TypeVar_1 , \MetaType' \equal \tbullet[\IntVar] \TypeVar_2\}}
} 
\]
It follows from  induction hypothesis that 
$\wte{\appTS{\substTE }{\TypeContextD},
\varX: \appTS{\substTE}{\TypeVar} }
{\Expression}{\appTS{\substTE}{\MetaType'}}$.
 Since $\substTE \models 
  \{\TypeVar \equal \tbullet[\IntVar] \TypeVar_1 , 
  \MetaType' \equal \tbullet[\IntVar] \TypeVar_2\}$,  we also have that 
$\wte{\appTS{\substTE }{\TypeContextD},
\varX: \tbullet[\appTS{\substTE}{\IntVar}] \appTS{\substTE}{\TypeVar_1} }{\Expression}{\tbullet[\appTS{\substTE}{\IntVar}]\appTS{\substTE}{\TypeVar_2}}
$.  
We can, then, apply $\defrule{\introarrow}$ of $\lambdabbminus$ 
to conclude that 
  $\wte{\appTS{\substTE }{\TypeContextD}}
  {\Fun \varX \Expression}
  {\tbullet[\appTS{\substTE}{\IntVar}]
  (\appTS{\substTE }{\TypeVar_1} \arrow \appTS{\substTE}{\TypeVar_2})}$.
  It follows from the fact that
  $\substTE$ is a type substitution that 
  $\appTS{\substTE}{\MetaType} =
  \tbullet[\appTS{\substTE}{\IntVar}](\appTS{\substTE }{\TypeVar_1} \arrow \appTS{\substTE}{\TypeVar_2})
 $ is a type.
  \end{proof}

\begin{theorem}[Completeness of Constraint Typing]
\label{theorem:completenessconstraint}
Let
 $\TypeContext = \var_1: \Type_1, \ldots, \var_m :\Type_m$  
  and 
  $\TypeContextD = \var_1: \TypeVar_1, \ldots, \var_m :\TypeVar_m$.
If $\wte{\TypeContext}{\Expression}{\Type}$ in $\lambdabbminus$
then there are  $\MetaType$ and $\setC$ such that

\begin{enumerate}

\item $\typec{\TypeContextD}{\Expression}{\MetaType}{\setC}$

\item 
there exists  a   type substitution  
$\substTE \supseteq \mapTS{\TypeVar_1}{\Type_1}{\TypeVar_m}{\Type_m}$ such that
$\substTE \models \setC$ and
$\appTS{\substTE}{\MetaType} \treeEq \Type$
and
$\dom (\substTE) \setminus \{\TypeVar_1, \ldots, \TypeVar_m\}$
is the set of fresh variables in the derivation of
$\typec{\TypeContextD}{\Expression}{\MetaType}{\setC}$.
\end{enumerate}

\end{theorem}

\begin{proof} 
We proceed by induction on the derivation.
Suppose the last typing rule in the derivation is:
\[
\inferrule[\defrule\axiom]{
        \varX_i: \Type_i \in \TypeContext
        }{
        \wte{\TypeContext}{\varX_i}{\bullet^n \Type_i}
        }
\]
 Then,  $
 \typec{\TypeContextD}{\varX_i}{\tbullet^\IntVar \TypeVar_i}{\emptyset}$.
  We define
  $\substTE = 
\mapTS{\TypeVar_1}{\Type_1}{\TypeVar_m}{\Type_m}
\cup    \singleTS{N}{n}$.
  It is easy to see that
  $\appTS{\substTE }{\TypeContextD} = \TypeContext$ and
  $\Type \treeEq \appTS{\substTE} {\tbullet[N] \TypeVar_i}$.

Suppose the last typing rule in the derivation is:
\[
\inferrule[\defrule\elimarrow]{
          \wte{\TypeContext}{\Expression_1}{\tbullet[n](\Type \arrow \TypeS)}
          \ \ \ 
          \wte{\TypeContext}{\Expression_2}{\tbullet[n]\Type}
        }{
        \wte{\TypeContext}{\Expression_1\Expression_2}{\tbullet[n]\TypeS}
        }
        \]
  
 By induction hypotheses,
 
 \begin{itemize}
 \item 
 $
  \typec{\TypeContextD}{\Expression_1}{\MetaType_1}{\setC_1}
  $ and
$  \typec{\TypeContextD}{\Expression_2}{\MetaType_2}{\setC_2}$,
\item 
$\substTE_1 \models \setC_1$ and 
      $\substTE_2 \models \setC_2$ and 
      $\appTS{\substTE_1 }{\MetaType_1} \treeEq
      \tbullet^n(\Type \arrow
      \TypeS)$ and 
      $\appTS{\substTE_2 }{\MetaType_1} \treeEq\tbullet^n\TypeS$
     for 
$
\substTE_1, \substTE_2 \supseteq 
\mapTS{\TypeVar_1}{\Type_1}{\TypeVar_m}{\Type_m}
$,
\item 
      $\dom (\substTE_1) \setminus 
      \{\TypeVar_1, \ldots, \TypeVar_m\}$ and  $\dom (\substTE_2) \setminus 
      \{\TypeVar_1, \ldots, \TypeVar_m\}$
are the sets of fresh variables in the derivation of
 $
  \typec{\TypeContextD}{\Expression_1}{\MetaType_1}{\setC_1}
  $
 and 
$  \typec{\TypeContextD}{\Expression_2}{\MetaType_2}{\setC_2}$, respectively.

\end{itemize}

It follows from the latter that 
 $\substTE_1 \cup \substTE_2$ is a function
 and we can define the substitution $\substTE$ as $
      \update{\update{\update{(\substTE_1 \cup \substTE_2)}{N}{n}
      }{\TypeVar_1}{\Type} }
      {\TypeVar_2}{\TypeS}$.         
\end{proof}

\subsection{Unification Algorithm}
\label{subsection:unification}

\cref{algorithm:unification} solves the unification problem.
The input consists of a set $\setC$ of  constraints 
and an extra argument $\setD$ that keeps tracks of the
``visited equality constraints'' and guarantees termination.

\begin{algorithm}
\SetKwInOut{Input}{input}\SetKwInOut{Output}{output}
\LinesNumbered
\SetKw{where}{{\small where}}
\SetKwHangingKw{sea}{{let}}
\SetKw{and}{{\small and}}
\SetKw{Return}{{\small return}}
\SetKw{Disj}{{\small or}}

 \Fn{$\Unify  \setC \setD $}{
\Input{$\setC$ is a set of constraints and  $\setD$ is the set of visited
equality  constraints
}
\Output{$\{ \langle \substT, \setE \rangle \mid \substT \cup \substE \models \setC \ \forall \substE \models \setE \}$
}

\If{$\setC$ is substitutional and simple  \label{algo:unif:basecase}}
   {
   $\assign{\setE}{\intc{\setC} \cup \Guards{\setC}}$ \label{algo:unif:guards}\;
   \leIf{ $\setE$ has a solution of non-negative integers \label{algo:unif:local}} 
      {\Return{$\{ \langle \substT_{\setC},   \setE \rangle  \}$}}
   {\Return{$\emptyset$} }
  }
\ElseIf{$\MetaType \equal \MetaTypeS \in \setC$ \label{algo:notsubst}}
{$\assign{\setC' }{\setC \setminus \{ \MetaType \equal \MetaTypeS \}}$; \ \ 
 $\assign{\setD' }{\setD \cup  \{ \MetaType \equal \MetaTypeS \}}$


\lCase{$\MetaTypeS =\MetaType$}
      {$\Unify{\setC'}{\setD'} $ \tcc*[f]{ $\MetaType \equal \MetaType$}}

\Case{ $\MetaType =  \TypeVar$ \and{$ \TypeVar \equal \MetaTypeS' \in \setC'$}
\tcc*[f]{ $\TypeVar \equal \MetaTypeS$ and $\TypeVar \equal \MetaTypeS'$}  \label{algo:unif:case3}}
{  \lIf{$\MetaType \equal \MetaTypeS \not \in \setD $} 
{ \Return{$\Unify {\setC' \cup \{  \MetaTypeS\equal \MetaTypeS'\}}{\setD'}$}} 
    \lElse{\Return{$\Unify  {\setC'} {\setD'}$}}
} 


\Case{
$ \MetaType =\MetaType_1 \binaryconstructor \MetaType_2$ \and{ $\MetaTypeS =\MetaTypeS_1 \binaryconstructor \MetaTypeS_2$}
\tcc*[f]{$\MetaType_1 \binaryconstructor \MetaType_2 \equal\MetaTypeS_1 \binaryconstructor \MetaTypeS_2$}}
{ \Return{$ \Unify {\setC' \cup \{ \MetaType_1 \equal \MetaTypeS_1 , \MetaType_2 \equal \MetaTypeS_2\}} {\setD'}$
}}


\Case{
$ \MetaType =\TypeVar$ \and{ $\MetaTypeS = \bullet^\IntExp \TypeVar'$}
\and{$\TypeVar \equal  \bullet^\IntExp \TypeVar' \not \in \setD$}
\tcc*[f]{$\TypeVar \equal \bullet^\IntExp \TypeVar' $}  \label{algo:unif:casevars}}
{
$\assign{\setCone}{ \setC' \cup \{\TypeVar \equal \TypeVar', \IntExp \equal 0 \}} $; 
\tcc*[f]{1st option} \label{algo:unif:varone}
 
$ \assign{\setCtwo}{\setC' \cup 
\{ \TypeVar \equal \bullet^{\IntExp} \TypeVar', 0 \menor  \IntExp  \} } $
\tcc*[f]{2nd option} \label{algo:unif:vartwo}

\Return{$  \Unify { \setCone} {\setD'} \cup \Unify {\setCtwo} {\setD'} $
\label{algo:unif:union3} }
}

\Case{$\MetaType = \tbullet^{\IntExp} \TypeVar$ \and{ $\MetaTypeS =\tbullet^{\IntExp'} \TypeVar' $} \label{algo:unif:case4}
 \tcc*[f]{ $\bullet^{\IntExp} \TypeVar \equal \tbullet^{\IntExp'} \TypeVar' $ }}
{
$\assign{\setCone}{ \setC' \cup \{ \TypeVar \equal   \bullet^{\IntExp' - \IntExp} \TypeVar', 
\IntExp' \mayor \IntExp, \IntExp \mayor 0 \}} $; 
\tcc*[f]{1st option} \label{algo:unif:fstoption}
 
$ \assign{\setCtwo}{\setC' \cup 
\{ \TypeVar' \equal \bullet^{\IntExp - \IntExp'} \TypeVar,
 \IntExp' \menor \IntExp, \IntExp' \mayor 0 \} } $
\tcc*[f]{2nd option} \label{algo:unif:sndoption}

\Return{$  \Unify { \setCone} {\setD'} \cup \Unify {\setCtwo} {\setD'} $
\label{algo:unif:union2} }
}

 \Case{$\MetaType = \tbullet^{\IntExp}\TypeVar$ \and{$\MetaTypeS = \tbullet^{\IntExp'} \MetaTypeS'$}
 \and{$\isNatOpBox{\MetaTypeS'}$}
 \tcc*[f]{$\tbullet^{\IntExp}\TypeVar  \equal \tbullet^{\IntExp'} \MetaTypeS'$} \label{algo:unif:case5}}
      {
      \Return{$
 \Unify {\setC' \cup \{ \TypeVar \equal   {\tbullet^{\IntExp' - \IntExp}\MetaTypeS'}, 
 \IntExp' \mayor \IntExp, \IntExp \mayor 0  \}} {\setD'}$ 
}}
%
%
\Case{$\MetaType =\tbullet^{\IntExp} \MetaType'$ 
 \and{$\MetaTypeS =  \tbullet^{\IntExp'} \MetaTypeS'$}
 \and{$\equalcons{\MetaType'}{\MetaTypeS'}$} \label{algo:unif:case6}
\tcc*[f]{$\tbullet^{\IntExp} \MetaType' \equal \tbullet^{\IntExp'} \MetaTypeS'$}}
{ \Return{$
 \Unify {\setC' \cup \{ \MetaType' \equal \MetaTypeS',\IntExp \equal \IntExp' \} }{\setD'}$
}}

\Case{$\MetaType = \tbullet^{\IntExp} \MetaType' $
\and{$\equalcons{\MetaType'}{\MetaTypeS}$ \label{algo:unif:case7}}
\tcc*[f]{ $\tbullet^{\IntExp}\MetaType' \equal \MetaTypeS $}}
{ \Return{$
 \Unify {\setC' \cup \{ \MetaType' \equal\MetaTypeS, \IntExp \equal 0 \} }{\setD'}$
} }


\Case{$\MetaType \equal  \MetaTypeS \not \in \setD$ \tcc*[f]{$\MetaType \equal \MetaTypeS$ } 
\label{algo:unif:swap}}
{ \Return{$ \Unify {\setC' \cup \{ \MetaTypeS \equal \MetaType
\} }
           {\setD' } $ \tcc*[f]{swap  metatypes (symmetry) }}
}
\lOther{\Return{$\emptyset$} \label{algo:unif:last}\tcc*[f]{failure}
}
}}

    \caption{Unification Algorithm}
\label{algorithm:unification}
\end{algorithm} 

The base case  on Line \ref{algo:unif:basecase}  is  
when the set $\setC = \eqc{\setC} \cup \intc{\setC}$ is
substitutional and simple (see Definitions \ref{definition:constraints},  \ref{definition:substitutional} and \ref{definition:Calwaysnormal}). 
The set $\intc{\setC}$ guarantee that the exponents of the $\bullet$'s
are positive. 
 \cref{algo:unif:guards} invokes the function $\guards$ which creates a set of constraints
 to ensure guardedness. 
Line \ref{algo:unif:local} checks whether $\setE= \intc{\setC} \cup \Guards{\setC}$ 
has a solution.
Altogether, this implies that $\substE \circ \substT_{\setC}$ is guarded and positive
for all natural $\substE \models \setE$.
The function $\guards$ is defined in
   \cref{algorithm:guards}  (\cref{appendix:local}).
In order to check that 
$\setE $ has a solution of non-negative integers (Line \ref{algo:unif:local}),
we can use  any    algorithm for linear integer programming~\cite{opac-b1095053}.

If  $\setC$  is neither substitutional nor simple,  we examine the equations in $\eqc{\setC}$ 
(\cref{algo:notsubst} onwards).

The case that eliminates two equality constraints for the same variable
(\cref{algo:unif:case3})
removes  $\TypeVar \equal \MetaTypeS$ from $\setC$
 but it adds $\MetaTypeS\equal \MetaTypeS'$. 
 If $\setD$ already contains 
  $\TypeVar \equal \MetaTypeS$,  it means that  $\TypeVar \equal \MetaTypeS$ has already been processed
  and  it will not  be added again  avoiding non-termination.
The unification algorithm for the
 simply typed lambda calculus solves the problem of termination 
 in this  case 
 by reducing the number of variables, i.e. checks 
 if $\TypeVar \not \in \MetaTypeS$ and then, substitutes  
  the variable $\TypeVar$ by $\MetaTypeS$
   in the remaining set of constraints.
With recursive types, however, we do not perform the occur check
and  the number of variables may not decrease
since  
   the variable $\TypeVar$ may not disappear after substituting
   $\TypeVar $ by $\MetaTypeS$ 
   (because $\TypeVar$ occurs in 
   $\MetaTypeS$).
  In order to decrease the number of variables, we could perhaps  substitute 
 $\TypeVar$ 
 by  the solution of the recursive equation 
 $ \TypeVar = \MetaTypeS$ (which may be {\em infinite}). 
  But the problem of guaranteeing termination would still be
  present because  the meta-types in the constraints can be infinite
  and  the size   of the equality constraints may not decrease.
We would also have a similar problem with termination if we use multi-equations
and a rewrite relation instead of giving a function such as $\unify$ 
\cite{pottier-remy-emlti}.

 The function $\unify$ does not return only one solution
 but a set of solutions. 
 This is because there are  two ways of solving the equality constraint
 $\bullet^{\IntExp} X \equal  \bullet^{\IntExp'} X'$
 (Line  \ref{algo:unif:case4}).
If $\bullet^{\IntExp} X \equal  \bullet^{\IntExp'} X'$, we solve 
either $\IntExp \mayor \IntExp'$ and $X' \equal \bullet^{\IntExp-\IntExp'} X$ or $\IntExp' \mayor \IntExp$ and $X \equal \bullet^{\IntExp'-\IntExp} X'$.

 Line \ref{algo:unif:casevars} 
 is similar to Line \ref{algo:unif:case4}, In order to solve
  $\TypeVar \equal \bullet^\IntExp \TypeVar'$, we solve either 
 $\IntExp \equal 0$ and  $\TypeVar \equal \TypeVar'$
 or  $0 \menor \IntExp$ and $\TypeVar \equal\bullet^\IntExp \TypeVar'$. 
 This distinction is needed because 
  in the particular case when $\TypeVar= \TypeVar'$, 
  the solution on the first case is any meta-type while the solution on the second
  case is $\infinite$.
  Termination is guaranteed in this case by checking that the constraint 
  $\TypeVar \equal \bullet^\IntExp \TypeVar'$ has not been visited before.

  The case on Lines \ref{algo:unif:case5}
  uses the function $\isNatOpBox{\MetaType}$ which returns true iff 
  $\MetaType$ is either $\nat$ or $(\MetaTypeS_1 \binaryconstructor \MetaTypeS_2)$ 
  (see \cref{algorithm:test} in \cref{appendix:local}).  
  It is clear that if $\IntExp' < \IntExp$,
then the equation $\tbullet^{\IntExp}\TypeVar \equal \tbullet^{\IntExp'} \MetaTypeS'$
does not have a positive  solution, e.g.
 $\bullet^{5} \TypeVar \equal 
\bullet^{2} (\MetaType_1 \times \MetaType_2)$ then the solution is
$\{ \TypeVar \mapsto
\bullet^{-3} (\MetaType_1 \times \MetaType_2) \}$ which is not positive.

  The cases on Lines  \ref{algo:unif:case6} and \ref{algo:unif:case7}
   use the function $\equalcons{\MetaType}{\MetaTypeS}$ which returns true iff
   both types are $\nat$ or built from the same type constructor 
   (see \cref{algorithm:equalcons} in \cref{appendix:local}).

  The algorithm analyses the form of the equation $\MetaType \equal \MetaTypeS$
  and the cases are done in order which means that at 
  \cref{algo:unif:swap}, the equation does not have any of the forms described in the above cases
  but it could be that $\MetaTypeS \equal \MetaType$ does. For this, the algorithm  swaps the equation 
  and 
  calls $\unify$ again. Termination is guaranteed by checking that the constraint does not belong to
  $\setD$. 
 In the last case (\cref{algo:unif:last}),
the algorithm  returns the empty set of solutions when all the previous  cases have failed.

The {\em size} $\size{\setC}$ of a set of meta-type constraints
is  the sum of the 
number of type variables and type constructors 
 in  the left hand side of the equality constraints. 
 Since  the meta-types in $\setC$ are finite, the size is always finite.
We define
\[
\begin{array}{ll}
\subtypes(\setC) & = \{ \MetaType \mid
\MetaTypeS_1 \equal \MetaTypeS_2 \in \setC \mbox{ and }
\mbox{ either } \MetaTypeS_1 \mbox{ or } \MetaTypeS_2 
\mbox{ contains } \MetaType \}
\\\allconstraints (\setC) &  =
\{ \MetaType_1 \equal \MetaType_2 \mid 
\MetaType_1, \MetaType_2 \in \subtypes (\setC) \}
\end{array}
\]

\begin{theorem}[Termination]
\label{theorem:unificationtermination}
Let $\setC$ and $\setD$ be sets of constraints.
Then,  
$\Unify  {\setC}{\setD}$ terminates 
\end{theorem}

\begin{proof}
Let $\setC_0 = \setC$ and $\setD_0 = \setD$.
Suppose  towards a contradiction that 
there is an infinite sequence of recursive calls:
\begin{equation}
\label{recursivecalls}
\Unify{\setC_0}{\setD_0}, \Unify{\setC_1}{\setD_1},\Unify {\setC_2}{\setD_2}, \ldots
\end{equation}
This means that $\Unify {\setC_i}{\setD_i}$ recursively calls 
$\Unify {\setC_{i+1}}{\setD_{i+1}}$ 
for all $i \geq 0$.
From the definition of the algorithm, one can see that 
 $\setC_i \subseteq \allconstraints (\setC_0)$
and $\setD_i \subseteq \setD_{i+1} \subseteq\setD_0 \cup \allconstraints (\setC_0)$ for all $i \geq 0$.
We define now a measure that decreases with each recursive call:
\[ \recn{i}  = 
(\size {\allconstraints (\setC_0)} + \size{\setD_0} - \size {\setD_i},
\size  {\setC_i})
\]
It is not difficult to show that $\recn{i} > \recn{i+1}$ contradicting the fact that
the sequence in \cref{recursivecalls} is infinite.
\end{proof}

\begin{theorem}[Completeness of Unification]
\label{theorem:completenessunification}
Let $\substTE$ be a  type substitution.
If  $\substTE \models \setC $ 
then there exists 
$\langle \substT, \setE \rangle \in \Unify {\setC} {\setD}$
such that 
$\substTE = \substTE' \circ \substT \restriction_{\dom(\substTE)}$ and $\substTE' = \substT' \cup \substE'$ and
$\substE' \models \setE$.
\end{theorem}

\begin{proof}
By induction on the  number of recursive calls which is finite by \cref{theorem:unificationtermination}.
We do only two cases.

Suppose $\setC$ is substitutional and simple.
Since $\substT = \substT_{\setC}$ is the most general unifier,
 there exists $\substTE' = \substT' \cup \substE'$ such that   
 $\substTE = \substTE' \circ \substT_{\setC}$.
 Since $\substTE$ is guarded, so is $\substE' \circ \substT_{\setC}$
 and $\substE' \models \setE = \Guards{\setC}$.
 
 Suppose  the constraint of $\setC$ which is being processed is
  $\bullet^{\IntExp} X \equal \bullet^{\IntExp'} X'$. Then, 
 
 \[ \bullet^{\substTE(\IntExp)} \substTE(X) = 
 \bullet^{\substTE(\IntExp')} \substTE(X')
 \]
 
 Then, there are two cases:
 \begin{enumerate} 
 \item Case $\substTE(\IntExp') \geq \substTE(\IntExp)$
 and $\substTE(X) = 
 \bullet^{\substTE(\IntExp')- \substTE(\IntExp)} \substTE(X')$.
 Then,
 $\substTE \models \setC_1 $ where $\setC_1$ is the set of constraints of 
 Line \ref{algo:unif:fstoption}.
  It follows by induction hypothesis that
  there exists 
$\langle \substT, \setE \rangle \in \Unify {\setC_1} {\setD}$
such that 
$\substTE = \substTE' \circ \substT \restriction_{\dom(\substTE)}$ and $\substTE' = \substT' \cup \substE'$ and
$\substE' \models \setE$.
We have that
 $\langle \substT, \setE \rangle \in \Unify {\setC} {\setD}$
 because  $\Unify {\setC_1} {\setD} \subseteq \Unify {\setC} {\setD}$.
  
 \item Case 
 $\substTE(\IntExp') < \substTE(\IntExp)$
 and $\substTE(X) = 
 \bullet^{\substTE(\IntExp)-  \substTE(\IntExp')} \substTE(X')$.
  Then,
 $\substTE \models \setC_2 $ 
 where $\setC_2$ is the set of constraints of 
 Line \ref{algo:unif:sndoption}. The  proof of this case is 
  similar to the first one.
\qedhere 
 \end{enumerate}
\end{proof}

\begin{theorem}[Soundness of Unification]
\label{theorem:soundnessunification}
Let $\setC$ be always positive.
If $\langle \substT, \setE \rangle \in \Unify {\setC}{\setD}$ then

\begin{enumerate}

\item   \label{theorem:soundnessunification:one}
there exists a natural $\substE$ such that
$\substE \models \setE$ and,

\item \label{theorem:soundnessunification:three}
for all  natural  $\substE$ such that
 $\substE \models \setE$, 
 
 \begin{enumerate}[(i)]
 \item 
$\substT \cup \substE \models \setC$ and,

 \item 
$\substT \cup \substE$ is a  type substitution.

  \end{enumerate}

\end{enumerate}

\end{theorem}

\begin{proof}
By induction on the number of recursive calls  which is finite by \cref{theorem:unificationtermination}.
We do only two cases.

 Suppose first that $\setC$ is substitutional and simple. Then, $\substT = \substT_{\setC}$ and 
 $\setE = \intc{\setC} \cup \Guards{\setC} $.
 \cref{theorem:soundnessunification:one} follows from the check  on \cref{algo:unif:local}.
 Assume that $\substE$ is natural and $\substE \models \setE $. 
 Part (i) of \cref{theorem:soundnessunification:three} holds since
 $\substT_{\setC} \cup \substE \models \setC$.
 We also have that $\substE \models \intc{\setC}$ and 
 $\substE \models \Guards{\setC}$.
Since  $\setC$ is always positive, all the meta-types in 
$\substE(\eqc{\setC})$ are positive and hence, 
 $\substT_{\setC} \cup \substE$  is positive.
 We also have that
 $\substT_{\setC}\cup \substE$ is guarded because $\substE \models \Guards{\setC}$.
 \end{proof}

\subsection{Type Inference Algorithm}
\label{subsection:algorithm}

 \cref{table:typeinference} solves the type inference problem in $\lambdabb$.
 Its inputs are  a decidable function  $\cond$ from meta-types to booleans 
 and a closed expression $\Expression$ and 
 the output is  a finite set of meta-types that cover
 all the possible types of 
 $\Expression$, i.e. 
 any type  of $\Expression$ 
 is an instantiation of one of those meta-types.

 \cref{algo:generateconstraints} computes $\MetaType$ and $\setC$
 such that $\typec{\emptyset}{\Expression}{\MetaType}{\setC}$
 from the constraint typing rules of \cref{subsection:constraint}.
 \cref{algo:callunify} calls
 the unification algorithm. The first argument of $\unify$ is
  the set $\setC$    and the second one is 
 $\emptyset$ for the set of visited equality constraints.   
  \cref{check:infinite} checks  that the meta-type  $ \appTS{\substT}{\MetaType}$
  satisfies the property $\cond$.   
  We are interested in 
  certain properties   $\cond (\MetaType)$ given below:
  
  \begin{enumerate}
  
  \item $\truecond{\MetaType}$ which returns true for all $\MetaType$.
  
  \item $\equalinfinite{\MetaType}$ returns true iff  $\MetaType$ 
   is different from  $\infinite$.
   The function $\equalinfinite{\MetaType}$ is  defined in \cref{algorithm:infinite}.
   The inference algorithm can be used in this case for  
excluding   non-normalizing programs.

\item 
$\freeinfinite{\MetaType}$
is true iff  $\MetaType$ is $\infinite$-free.
The function 
$\freeinfinite{\MetaType}$
 is defined in \cref{algorithm:infinitefree}.
The inference algorithm can be used in this case for  
excluding 
programs that have L\`evy-Longo trees with  $\bot$.

\item $\tailfinite{\MetaType}$ which checks 
whether $\MetaType$ is tail finite
 The function $\tailfinite{\MetaType}$
is defined in \cref{algorithm:tailfinite}.
The inference algorithm can be used in this case for  
excluding 
 programs that  have B\"ohm trees with $\bot$.
\end{enumerate}

\begin{algorithm}
\SetKwInOut{Input}{input}\SetKwInOut{Output}{output}
\SetKwFunction{True}{true}
\SetKwFunction{False}{false}
\SetKw{where}{{\small where}}
\SetKwHangingKw{sea}{{let}}
\SetKw{and}{{\small and}}
\SetKw{Return}{{\small return}}
\SetKw{Disj}{{\small or}}
\LinesNumbered

\Fn{$\infertype (\cond, \Expression)$}{
\Input{$\cond$ is a decidable function from meta-types to booleans 
and $\Expression$ is a closed  expression}
\Output{set of pairs $\langle \MetaType, \setE \rangle$ such that
$\Type =  \appTS{\substE}{\MetaType}$ and 
$\vdash \Expression :\Type$ for all $\substE \models \setE$  }

Compute $\MetaType$ and $\setC$ such that
 $\typec{\emptyset}{\Expression}{\MetaType}{\setC}$ \label{algo:generateconstraints}

$\assign{\setT}{\emptyset}$

\ForEach{$\langle \substT, \setE \rangle \in \Unify{ \setC}{ \emptyset}$ \label{algo:callunify}}
        {\lIf {$\cond (\appTS{\substT}{\MetaType})$ \label{check:infinite}}
             { $\assign{\setT}{\setT \cup \langle \appTS{\substT}{\MetaType}, \setE \rangle $}}}
 \Return{$\setT$}
 }

\caption{Type Inference Algorithm}

\label{table:typeinference}

\end{algorithm}

\begin{theorem}[Completeness  of Type Inference]
\label{theorem:completenesstypeinference} 
Let $\vdash {\Expression}: {\Type}$ in $\lambdabb$. 
Then,

\begin{enumerate}
 
\item there exists 
 $\langle \MetaType, \setE \rangle \in \infertype (\pone, \Expression)$
 such that
 $\Type = \appTS{\substTE}{\MetaType}$ for some $\substTE = \substT \cup \substE$
 and 
 $\substE \models \setE$,

\item $\langle \MetaType, \setE \rangle \in \infertype (\ptwo, \Expression)$ 
if 
 $\Type\not=\infinite$,

\item  $\langle \MetaType, \setE \rangle \in \infertype (\pthree, \Expression)$ 
if  the L\`evy-Longo tree of  $\Expression$ has no $\bot$'s and

\item $\langle \MetaType, \setE \rangle \in \infertype (\pthree, \Expression)$ 
if  the B\"ohm  tree of  $\Expression$ has no $\bot$'s.

\end{enumerate}
\end{theorem}

\begin{proof}
The first part follows from 
\cref{lemma:lambdabbminus}, 
\cref{theorem:completenessconstraint}
and \cref{theorem:completenessunification}.
We show only the second part since the remaining parts can be proved similarly.
Suppose $\substTE(\MetaType) = \Type\not=\infinite$.
But this means that $\MetaType$ cannot be $\infinite$ and 
$\langle \MetaType, \setE \rangle \in \infertype (\ptwo, \Expression)$.
\end{proof}

\begin{theorem}[Soundness of Type Inference]
\label{theorem:soundnesstypeinference} 
 If $\langle \MetaType, \setE \rangle \in \infertype (\cond, \Expression)$
 then   
 
 \begin{enumerate}
 
\item \label{theorem:soundnesstype:zero}
$\cond (\MetaType)$ is true,
 
\item \label{theorem:soundnesstype:one}
there exists  a natural $\substE$ such that $\substE \models \setE$ and 
 \item 
\label{theorem:soundnesstype:two}
 for all natural $\substE $ such that $\substE \models \setE$,
 \begin{enumerate}[(i)]
 \item  \label{theorem:soundnesstype:twoone}
  $\appTS{\substE}{\MetaType } =  \Type$ is a type 
  and $\vdash \Expression:\Type$ in $\lambdabb$. 
  
 \item \label{theorem:soundnesstype:twotwo}
 If $\cond$ is $\ptwo$ then $\Type \not = \infinite$.
 
  \item \label{theorem:soundnesstype:twothree}
 If $\cond$ is $\pthree$ then $\Type$ is $\infinite$-free.
 
  \item \label{theorem:soundnesstype:twofour}
 If $\cond$ is $\ptwo$ then $\Type$ is tail finite.

\end{enumerate}

 \end{enumerate}
\end{theorem}

\begin{proof} 
Suppose $\langle \MetaType, \setE \rangle \in \infertype (\Expression)$.
It follows from the definition of \cref{table:typeinference}
that  $\MetaType = \appTS{\substT}{\MetaTypeS}$
for some  $\typec{\emptyset}{\Expression}{\MetaTypeS}{\setC}$
and $\langle \substT, \setE \rangle \in \Unify {\setC}{\emptyset}$.

\cref{theorem:soundnesstype:zero} follows from the definition of 
\cref{table:typeinference}.

\cref{theorem:soundnesstype:one}   follows from  \cref{theorem:soundnessunification:one} of 
\cref{theorem:soundnessunification} and from the fact that $\setC$ is always positive.
One can show by induction on the rules of  Table \ref{table:constraints}
that $\setC$ is  always positive. 

In order to prove  \cref{theorem:soundnesstype:two}, assume that 
$\substE$ is natural and $\substE \models \setE$ and take 
$\substTE = \substT \cup \substE$.
We only prove  \ref{theorem:soundnesstype:twoone} since 
Items \ref{theorem:soundnesstype:twotwo}, \ref{theorem:soundnesstype:twothree}
 and \ref{theorem:soundnesstype:twofour} follow from  
 Items \ref{theorem:soundnesstype:zero} and \ref{theorem:soundnesstype:twoone}.
It follows from \cref{theorem:soundnessunification:three} of  
\cref{theorem:soundnessunification}
that $\substTE \models \setC$ and 
$\substTE$ is a  type substitution.
By \cref{theorem:soundnessconstraint},
$\wte{ }{\Expression}
{\appTS{\substTE}{\MetaTypeS}}$ in $\lambdabbminus$ and 
 $\appTS{\substTE}{\MetaTypeS} = \appTS{\substE}{\MetaType} = \Type $ is a type.
By \cref{lemma:lambdabbminus}, we have that 
$\wte{}{\Expression}{\appTS{\substTE}{\MetaTypeS}}$ in $\lambdabb$.

\end{proof}

\subsection{Examples}
\label{subsection:examplestypeinference}

We now illustrate the  type inference algorithm through some examples.

Consider $\Expression = \lambda x. x$. 
The result of   $\infertype (\pone, \lambda x. x)$ is a set that contains only one meta-type:
\[
\begin{array}{ll}
\MetaTypeS = \bullet^{\IntVar} (\TypeVar_1 \rightarrow \bullet^{\IntVarM} \TypeVar_1)
\end{array}
\]
with the implicit set of integer constraints $\IntVar, \IntVarM \geq 0$.
This meta-type covers all solutions, i.e. any type of $\Expression = \lambda x.x$
in $\lambdabb$ 
is an instantiation of $\MetaTypeS$.
How is this meta-type obtained by the algorithm?
We can apply  the rule \refrule{\introarrow} of  \cref{table:constraints}
 which gives the type derivation
 $\typec{\emptyset}{\Expression}{\MetaType}{\setC}$
where $\MetaType = \bullet^{\IntVar} (\TypeVar_1 \arrow \TypeVar_2)$
and $\setC = \{ \TypeVar \equal  \bullet^{\IntVar} \TypeVar_1, 
                  \bullet^{\IntVarM} \TypeVar \equal  \bullet^{\IntVar} \TypeVar_2 \}$.
The result of $\Unify{\setC}{\emptyset}$ is
    $\substT= \{ \TypeVar \mapsto   \bullet^{\IntVar} \TypeVar_1, 
                  \TypeVar_2 \mapsto   \bullet^{\IntVarM} \TypeVar_1 \}$
                  and
$\appTS{\substT}{\MetaType} =                  
     \bullet^{\IntVar} (\TypeVar_1 \rightarrow \bullet^{\IntVarM} \TypeVar_1) = \MetaTypeS$.

Consider now  $\lambda x. x x$.  Then, the result of 
$\infertype (\pone, \lambda x. x x)$  is  a set that contains two meta-types.
The first meta-type is  $\MetaTypeS_1 = \bullet^{\IntVara} (\TypeVarA \rightarrow \bullet^{\IntVarc - \IntVara} \TypeVarD)$
where $\TypeVarA$ should satisfy the recursive equation:
\begin{equation}
\label{equation:recursivetypeone}
\TypeVarA \equal \bullet^{\IntVarc -(\IntVara + \IntVard)} (\bullet^{\IntVara + \IntVare - \IntVarc} \TypeVarA\rightarrow\TypeVarD)
\end{equation}
and the set  of integer constraints for $\MetaTypeS_1$ is:
\begin{equation}
\label{equation:constraintsone}
\setE_1 =  \{\IntVara + \IntVare \geq \IntVarc \geq \IntVara + \IntVard,
\IntVare -\IntVard \geq 1 \}
\end{equation}
The constraints 
$ \IntVara + \IntVare \geq \IntVarc \geq \IntVara + \IntVard$
 guarantee that the exponents of the $\bullet$'s are all positive.
 The constraint $\IntVare -\IntVard \geq 1$
  forces  the recursive type to be guarded.

The second meta-type is  $\MetaTypeS_2 = \bullet^{\IntVara} (\bullet^{\IntVarc - (\IntVara+ \IntVare)} \TypeVarC \rightarrow \bullet^{\IntVarc - \IntVara} \TypeVarD)$
where $\TypeVarC$ should satisfy the recursive equation:
\begin{equation}
\label{equation:recursivetypetwo}
\TypeVarC \equal \bullet^{\IntVare - \IntVard} ( \TypeVarC\rightarrow\TypeVarD)
\end{equation}
and the set of integer constraints for $\MetaTypeS_2$ is:
\begin{equation}
\label{equation:setconstraintstwo}
\setE_2 = \{ \IntVarc > \IntVara + \IntVare, \IntVare \geq \IntVard, \IntVare -\IntVard \geq 1 \}
\end{equation}
The constraints 
$ \IntVarc > \IntVara + \IntVare$ and $ \IntVare \geq \IntVard$
 guarantee that the exponents of the $\bullet$'s are all positive.
 The constraint $\IntVare -\IntVard \geq 1$
  forces  the recursive type to be guarded.

\subsection{Typability and Type Checking}
\label{subsection:typability}

Typability (finding out if the program is typable or not) 
in $\lambdabb$ can be  solved by checking whether 
$\infertype (\pone, \Expression) \not = \emptyset$.
By \cref{theorem:completenesstypeinference}, if 
$\infertype (\pone, \Expression) = \emptyset$ then there is no 
$\Type$ such that
$\vdash \Expression: \Type$ in $\lambdabb$.
By \cref{theorem:completenesstypeinference}, if 
$\infertype (\Expression) \not = \emptyset$ then there exists 
$\Type$ such that 
$\vdash \Expression: \Type$ in $\lambdabb$.

Type checking is usually an easier problem than type inference.
In the case of  $\lambdabb$ (given an an expression $\Expression$ and a type $\Type$ check if $\vdash \Expression :\Type$), type checking 
can easily be solved  
by inferring the (finite) set of meta-types for $\Expression$ and
checking whether one of these meta-types  unifies with $\Type$.

The type inference algorithm is exponential in the size of the input
because of the unification algorithm which branches (case  \ref{algo:unif:case4} of 
\cref{algorithm:unification}). 
If we are only interested in typability (and not type inference), then the complexity could be reduced to 
$\NP$ by ``guessing the solutions''.
Instead of giving the union of the results
 in \cref{algo:unif:union2} of \cref{algorithm:unification}, we could 
 choose one of the options non-deterministically and 
 call $\unify$ only ``once''.

\section{Related Work}
\label{section:relatedwork}

This paper is an improvement over past typed lambda calculi with a
temporal modal operator  like $\bullet$ in two respects.
Firstly, we do not need any subtyping relation as in
\cite{Nakano00:lics} and secondly  programs are not cluttered
with constructs for the introduction and elimination of individuals of
type $\bullet$  as in
\cite{KrishnaswamiB11,severidevriesICFP2012,KrishnaswamiBH12,AM13,DBLP:conf/popl/CaveFPP14,CBGB15,DBLP:conf/fossacs/BizjakGCMB16,Guatto2018}.

Another type-based approach for ensuring productivity 
are sized types \cite{DBLP:conf/popl/HughesPS96}.
Type systems using size types do not always  have neat properties: 
 strong normalisation
 is gained  by contracting the fixed point operator
 inside a case and they lack the property of
  subject reduction \cite{DBLP:conf/lics/Sacchini13}.
Another disadvantage of size types is that they 
do not include negative occurrences of the recursion variable
\cite{DBLP:conf/popl/HughesPS96} which are useful for some applications \cite{DBLP:conf/esop/SvendsenB14}.

The  proof assistant Coq does not  ensure productivity through
typing but by means of  a syntactic guardedness condition
(the recursive calls should be guarded by
constructors) \cite{GimenezCasteranTutorialCoq,Coquand93}  
 which is somewhat restrictive since it rules out
 some interesting functions \cite{severidevriesICFP2012,CBGB15}.

A sound but not complete type inference algorithm for Nakano's
type system is presented in \cite{Rowe2012}.
This means that the expressions typable by the algorithm
are also typable in Nakano's system but the converse
is not true.
Though this algorithm is tractable, it is not clear
to which type system it corresponds.


\section{Conclusions and Future Work}
\label{section:conclusions}

%
%
%

The typability problem  is trivial in $\lambda \mu$
because all expressions are typable 
using  $\mu \TypeVar. (\TypeVar \rightarrow \TypeVar)$.
In $\lambdabb$, this problem turns out to be  interesting  
because it  gives us
 a way of filtering programs that do not satisfy certain properties, i.e.
 normalization, having a L\`evy-Longo or a B\"ohm tree without $\bot$.
 It is also challenging because it involves the generation of 
 integer constraints.



Due to the high complexity of the type inference algorithm, it will be important
to  find optimization techniques (heuristics, use of concurrency, etc) to make it feasible.
 
%

 It will be interesting to investigate the interaction 
 of this modal  operator  with dependent types.
As observed in Section \ref{section:calculus}, our type system is not
closed under $\eta$-reduction and this property may be useful for establishing a better
link with the denotational semantics.  
We leave the challenge of 
attaining a normalising and decidable type system  closed under  $\beta\eta$-reduction
for future research. 
It will also be interesting to investigate ways of extending $\lambdab$ to be able
to type the programs $\get$ and $\take$ of Section \ref{section:untypable}.

\paragraph{Acknowledgments.} 
I am  grateful to Mariangiola Dezani-Ciancaglini,
the anonymous reviewers of FOSSACS 2017 and LMCS for their useful
suggestions, which led to substantial improvements.


\newpage

\appendix

\section{Proof of \cref{theorem:soundness}}
\label{appendix:normalisation}

\begin{lemma}\label{lem:B}
\begin{enumerate}
\item\label{lemma:closedunderreductionandexpansion}
Let $\Expression\red\Expression'$. 
Then\comma 
$\Expression\in\indti\Type\ind$ iff $\Expression'\in\indti\Type\ind$
for all $\ind \in \natset$ and type $\Type$.
\item \label{lemma:constants}
If $\Constant:\Type \in \ConsCtx$, then
 $\Constant\in \bigcap_{\ind \in \natset} \indti{\Type}{\ind}$.
\end{enumerate}
\end{lemma}
\begin{proof} (\cref{lemma:closedunderreductionandexpansion}). 
By induction on $(\ind, \typerank(\Type))$.

(\cref{lemma:constants}).  Using the definition of type intepretation.
\end{proof}

\begin{lemma}\label{lem:C}
If $\funsubst \modelsi \TypeContext$,  then
$\funsubst \modelsj \TypeContext$ for all $\indj \leq \ind$.
\end{lemma}
\begin{proof}
It  follows from \ri{\cref{lem:A}}{\cref{lemma:monotonicityofinterpretation}}.
\end{proof}

{\bf Proof of ~\cref{theorem:soundness}}.
\begin{proof} 
We prove that $\TypeContext \modelsi \Expression:\Type$ for all $\ind \in \mathbb{N}$
by induction on  $\wte{\TypeContext}\Expression\Type$.

\mycase{Rule  \refrule{\const}}

\noindent 
It   follows from \ri{\cref{lem:B}}{\cref{lemma:constants}}.

\mycase{Rule \refrule\introbullet}

\noindent 
Suppose $i=0$. Then\comma 
\[
\funsubst(\Expression) \in \indti{\Type} {0}= \EE
\]
Suppose $i >0$ and  $\funsubst \modelsi \TypeContext$.
If $x : \TypeS \in \TypeContext$ then 
$\funsubst(x) \in \indti{\TypeS} {\ind} \subseteq  \indti{\TypeS} {\ind-1}$ 
by \ri{\cref{lem:A}}{\cref{lemma:monotonicityofinterpretation}}.
Hence, $\funsubst \models_{i-1} \TypeContext$.
By induction hypothesis\comma 
  $\TypeContext \modelsj \Expression:\Type$ for all $j \in \mathbb{N}$.   
Hence\comma 
$\funsubst(\Expression) \in \indti{\Type} {\ind-1}$
and 
\[
\funsubst(\Expression) \in \indti{\Type} {\ind-1}= \indti{\tbullet{\Type}}{\ind}
\]

\mycase{Rule \refrule\elimarrow}
 
 \noindent The derivations   ends with the rule:   
  \[
  \inferrule{
    \wte{\TypeContext}{\Expression_1}{\tbullet[n](\TypeS \arrow \Type)}
    \\
    \wte{\TypeContext}{\Expression_2}{\tbullet[n]\TypeS}
  }{
    \wte{\TypeContext}{\Expression_1\Expression_2}{\tbullet[n]\Type}
  }
  \]
  with 
  $\Expression = \Expression_1\Expression_2$.
By induction hypothesis\comma for all $\ind \in \natset$
\begin{equation}
\label{equation:inductiohypothesisoperator}
\TypeContext \modelsi \Expression_1: \tbullet[n](\TypeS \arrow \Type)
\end{equation}
\begin{equation}
\label{equation:inductiohypothesisargument}
\TypeContext \modelsi\Expression_2: \tbullet[n]\TypeS
\end{equation}
We have two cases:
\begin{enumerate}
\item Case $\ind < n$.
By \ri{\cref{lemma:interpretationofmanybullets}}{\cref{lemma:interpretationofmanybullets1}}\comma
$\indti{\tbullet[n] \Type} {\ind} = \EE$. We trivially  get  
\[
\funsubst (\Expression_1 \Expression_2) \in
\indti{\tbullet[n] \Type} {\ind}
\]
\item Case $\ind \geq n$. Suppose that $\funsubst \modelsi \TypeContext$. 
 
\begin{equation}
\label{equation:interpretationexpone}
\begin{array}{lll}
\funsubst(\Expression_1) & \in & \indti{\tbullet[n](\TypeS \arrow \Type)} \ind 
\mbox{ by \eqref{equation:inductiohypothesisoperator} } \\
              & =   & \indti{(\TypeS \arrow \Type)} {\ind-n} 
\mbox{ by \ri{\cref{lemma:interpretationofmanybullets}}{\cref{lemma:interpretationofmanybullets2}}} 
\end{array}
\end{equation}
\begin{equation}
\label{equation:interpretationexptwo}
\begin{array}{lll}
\funsubst(\Expression_2) & \in & \indti{\tbullet[n]\TypeS} \ind
\mbox{ by \eqref{equation:inductiohypothesisargument}} \\
   & =   & \indti{\TypeS} {\ind-n}
   \mbox{ by \ri{\cref{lemma:interpretationofmanybullets}}{\cref{lemma:interpretationofmanybullets2}}} 
   \end{array}
\end{equation}
By Definition of $\indti{(\TypeS \arrow \Type)}{\ind-n} $ and 
(\ref{equation:interpretationexpone})\comma 
there are 
two possibilities:
\begin{enumerate}
\item Case $\funsubst(\Expression_1) \in  \WNVAR $. Then\comma
\begin{equation}
\label{equation:applicationsoundness}
\funsubst(\Expression_1 \Expression_2) =
\funsubst(\Expression_1) \funsubst(\Expression_2)
\red^{*} \Context[\varX] \funsubst(\Expression_2)
\end{equation}
Hence\comma 
\[
\begin{array}{lll}
\funsubst(\Expression_1 \Expression_2) & \in \WNVAR & \mbox{by (\ref{equation:applicationsoundness})}  \\
&  \subseteq \indti{\tbullet[n] \Type}{\ind} &\mbox{ by \ri{\cref{lem:A}}{\cref{lemma:wnvar}}}.
\end{array}
\]
\item Case $ \funsubst(\Expression_1) \red^{*} \lambda x. \Expression' $
or $ \funsubst(\Expression_1) \red^{*} \Context[\Constant]$.
We also have that
\[
 \funsubst(\Expression_1)\Expression'' \in \indti{\Type}{\ind-n}  \ \ 
 \forall \Expression'' \in \indti{\TypeS}{\ind-n} 
 \]
 In particular 
 (\ref{equation:interpretationexptwo})\comma 
  implies 
\[
\funsubst (\Expression_1 \Expression_2) = \funsubst(\Expression_1) \funsubst(\Expression_2)
 \in \indti{\Type} {\ind -n}
\]
Since $\indti{\Type} {\ind -n}= \indti{\tbullet[n] \Type}{\ind}$ by  \ri{\cref{lemma:interpretationofmanybullets}}{\cref{lemma:interpretationofmanybullets2}}\comma we are done.
\end{enumerate}
\end{enumerate}



\mycase{Rule \refrule\introarrow}

\noindent The derivation   ends with the rule:   
    \[
    \inferrule{
      \wte{\TypeContext,\varX: \tbullet[n]\Type }{\Expression}{\tbullet[n]\TypeS}
    }{
      \wte{\TypeContext}{\Fun \varX \Expression}{\tbullet[n](\Type \arrow \TypeS)}
    }
    \]
 By induction hypothesis\comma for all $\ind \in \natset$ 
 \begin{equation}
\label{equation:inductionhypothesisbodyabstraction}
 \TypeContext,\varX: \tbullet[n]\Type \modelsi {\Expression} : {\tbullet[n]\TypeS}
 \end{equation}    
    We have two cases:   
    \begin{enumerate}    
    \item Case $\ind < n$. By \ri{\cref{lemma:interpretationofmanybullets}}{\cref{lemma:interpretationofmanybullets1}}\comma
$\indti {\tbullet[n](\Type \arrow \TypeS)} \ind = \EE$. We trivially 
 get 
     \[
        \funsubst (\Fun \varX \Expression) \in             \indti {\tbullet[n](\Type \arrow \TypeS)} \ind 
           \]    
    \item  Case $\ind \geq n$.  Suppose that $\funsubst \modelsi \TypeContext$. 
    By \ri{\cref{lemma:interpretationofmanybullets}}{\cref{lemma:interpretationofmanybullets2}}\comma
     it is enough  
    to prove 
    that 
    \[
    \funsubst (\Fun \varX \Expression) \in \indti {\Type \arrow \TypeS}{\ind - n}
    \]
    For this\comma suppose   
    $\ExpressionF \in  \indti{\Type}{\indj}$
    for $\indj \leq \ind -n$.
    We consider the substitution function defined as 
    $\funsubst_0 =  \funsubst \cup \{ (\varX, \ExpressionF) \}$.
    We have that  
    \begin{equation}
    \label{equation:funsubstj+n}
    \funsubst_0\models_{j+n} \TypeContext, \varX :\tbullet[n]\Type 
    \end{equation}
    because
    \begin{enumerate}
    \item $\funsubst_0(\varX) = 
    \ExpressionF \in  \indti{\Type}{\indj} = \indti{\tbullet[n]
    \Type}{\indj+n}$
    by \ri{\cref{lemma:interpretationofmanybullets}}{\cref{lemma:interpretationofmanybullets2}}.
    
    \item $\funsubst_0 \models_{\indj+n} 
    \TypeContext$
    by \cref{lem:C}
    and the fact that $\funsubst_0 \modelsi \TypeContext$.
    \end{enumerate}
  It follows from \eqref{equation:inductionhypothesisbodyabstraction} and
  \eqref{equation:funsubstj+n} that
  \begin{equation}
  \label{equation:rhozero}
  \funsubst_0 (\Expression) 
        \in \indti{\tbullet[n]\TypeS}{\indj+n}
  \end{equation}
   Therefore\comma we  obtain   
    \[\begin{array}{lll}
     (\Fun \varX \Expression) \ExpressionF \red
       \funsubst(\Expression)\subst{\ExpressionF}{\varX}
       = \funsubst_0 (\Expression) 
        & \in \indti{\tbullet[n]\TypeS}{\indj+n} & \mbox{by 
        (\ref{equation:rhozero}}) \\
        & = \indti{\TypeS}{\indj} & \mbox{by  \ri{\cref{lemma:interpretationofmanybullets}}{\cref{lemma:interpretationofmanybullets2}}}
        \end{array}
    \]    
      By \ri{\cref{lem:B}}{\cref{lemma:closedunderreductionandexpansion}}\comma we conclude
      \[\begin{array}{lll}
     (\Fun \varX \Expression) \ExpressionF 
        & \in \indti{\TypeS}{\indj}.
        \end{array}
        \qedhere
    \]       
    \end{enumerate}
\end{proof}

\newpage

\section{Auxiliary Functions for the Type Inference Algorithm}
\label{appendix:local}

This section gives several algorithms (functions) that are used in the
unification and type inference algorithms.

Let $\Subtrees \MetaType $ denote the set of subtrees of $\MetaType$ and 
$\Psubtrees \MetaType $ denote the set of proper subtrees of $\MetaType$.
If $\MetaType$ is a regular pseudo meta-type then,
$\Subtrees \MetaType $ and $\Psubtrees \MetaType $ are finite.

\subsection{Guardness}
\cref{algorithm:guards} defines the function $\guards$ that given  
a set $\setC$ of meta-type constraints returns a set $\setE$ of integers constraints
that guarantees guardedness.

\begin{algorithm}
\SetKwInOut{Input}{input}\SetKwInOut{Output}{output}
\SetKwFunction{True}{true}
\SetKwFunction{False}{false}
\SetKw{where}{{\small where}}
\SetKwHangingKw{sea}{{let}}
\SetKw{and}{{\small and}}
\SetKw{Return}{{\small return}}
\SetKw{Disj}{{\small or}}

\Fn{$\Guards \setC$}{
\Input{$\setC$ is a  set of substitutive constraints}
\Output{  $\Guards \setC = \setE$ such that
      $\substE \circ \substT_{\setC}$ is guarded iff
       $\substE \models \setE$ for all  natural $\substE$.
      }

\BlankLine

\Return{$\bigcup_{\TypeVar \in \dom (\substT_{\setC})} \Guards {\substT_{\setC}(\TypeVar)}$
}}

\Fn{$\Guards \MetaType$}{
\Input{$\MetaType$ regular pseudo meta-type}
\Output{$\Guards \MetaType = \setE$  such that
      $\substE(\MetaType)$ is guarded iff
       $\substE \models \setE$ for all  natural  $\substE$.
      }

\BlankLine

 $\assign{\setE}{\emptyset}$\;
 
\ForEach{$\MetaTypeS \in \Psubtrees \MetaType $ such that $\MetaTypeS \in \Psubtrees \MetaTypeS $ 
\tcc*[f]{$\MetaTypeS = \ldots \MetaTypeS \ldots $}}
{
$\assign{\setE}{\setE \cup \{ \IntExp \mayor 1 \mid \IntExp \in \Countb{\MetaTypeS}{\MetaTypeS}} \}$
}
     
\Return{$\setE$}}

\Fn{$\Countb{\MetaTypeS}{\MetaType}$}{
\Input{$\MetaTypeS$, $ \MetaType$ regular  pseudo meta-types}
\Output{
the set of $\IntExp_{1} + \ldots + \IntExp_{n}$ such that 
$\MetaType = \bullet^{\IntExp_{1}} (\ldots (\bullet^{\IntExp_{n}} \MetaTypeS) \ldots)$
} 

\BlankLine    

\lIf{$\MetaTypeS \not \in \Psubtrees {\MetaType}$}{\Return {$\emptyset$}}
\lCase{$\MetaType = \MetaTypeS$ } {\Return {$\{ 0 \}$}}
\lCase{$\MetaType = \bullet^{\IntExp} \MetaType'$ } 
{\Return {$\{ \IntExp + \IntExp' \mid \IntExp' \in \Countb{\MetaTypeS}{\MetaType'}\}$}}
\lCase{$\MetaType = (\MetaType_{1} \binaryconstructor \MetaType_{2})$}
{\Return{$\Countb{\MetaTypeS}{\MetaType_{1}}\cup  \Countb{\MetaTypeS}{\MetaType_{2}}$}}
}

  \caption{Integer constraints to guarantee guardedness}
\label{algorithm:guards}
\end{algorithm} 
%
 The function $\Countb{\MetaTypeS}{\MetaTypeS}$ 
 counts the number of bullets from the root of  
 $\MetaTypeS$ to each recursive occurrence of $\MetaTypeS$, i.e. it  gives the  
  set of $\IntExp_{1} + \ldots + \IntExp_{n}$ such that 
$\MetaTypeS = \bullet^{\IntExp_{1}} (\ldots (\bullet^{\IntExp_{n}} \MetaTypeS) \ldots)$
  In the figure below, $\Countb{\MetaTypeS}{\MetaTypeS}$ returns
  $\{ \IntExp_{1} + \IntExp_{2}+ \IntExp_{3}, \IntExp_{4} + \IntExp_{5}\}$.
  
\[
 \includegraphics[scale=0.1]{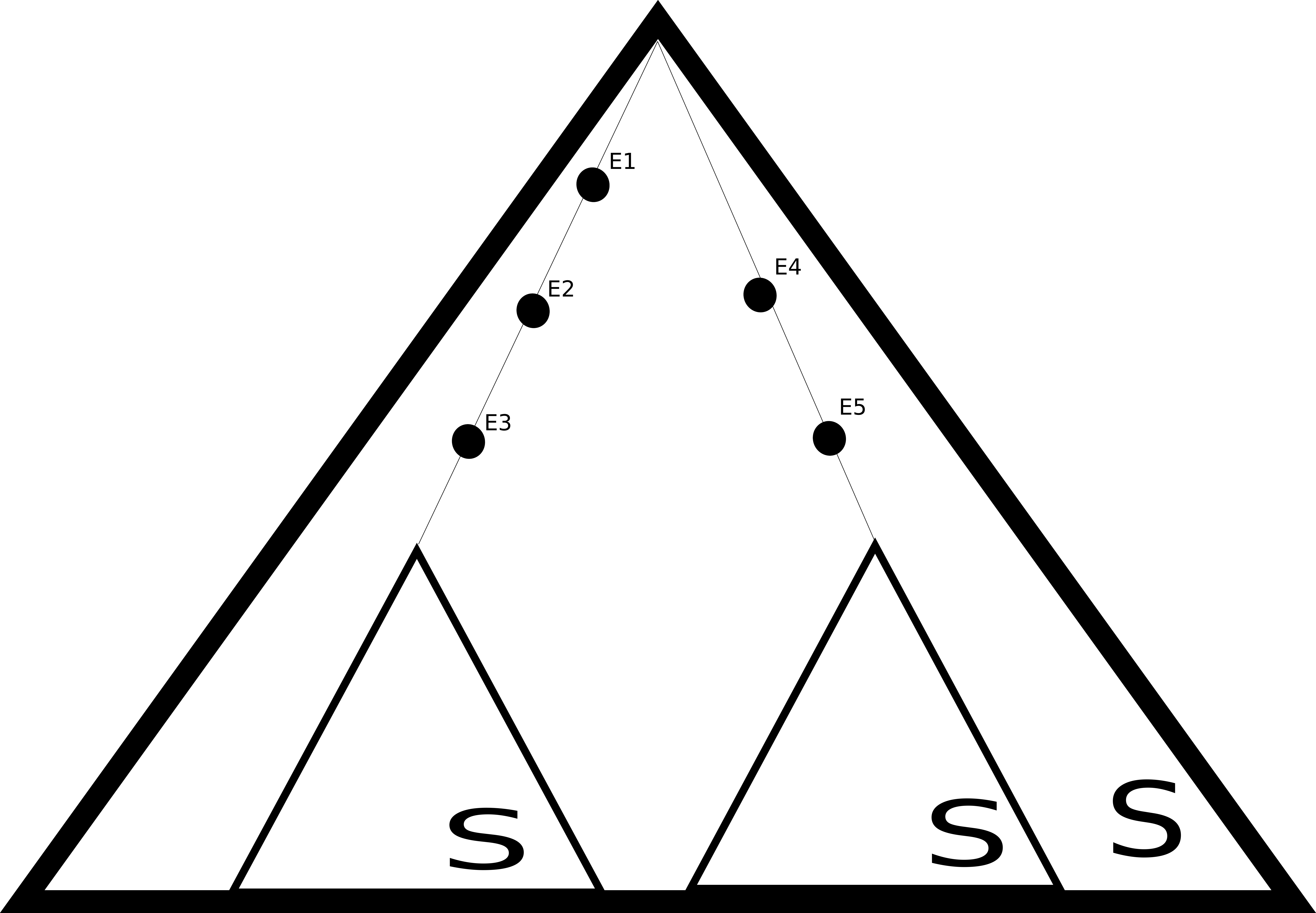}
\]
For example,

\begin{enumerate}

\item 
Let 
$\setC = \{ \TypeVar  \equal \bullet^{\IntVar} (\nat \arrow \bullet^{\IntVarM} \TypeVar) \}$.
The set $\setE =\{ \IntVar + \IntVarM \mayor 1 \}$ 
 guarantees that $\substT_{\setC}$ is guarded, i.e. 
$\substE \circ \substT_{\setC}$ is guarded for all $\substE$ such that
$\substE \models \setE$.

\item Let 
$ \setC = \{
\TypeVar_1 \equal  \bullet^{N_1} (\bullet^{N_2} \TypeVar_1 \rightarrow \bullet^{N_3} \TypeVar_2),  \ 
\TypeVar_2 \equal  \bullet^{M_1} (\bullet^{M_2} \TypeVar_2 \rightarrow \bullet^{M_3} \TypeVar_3), \
\TypeVar_3 \equal \bullet^{K_1} (\bullet^{K_2} \TypeVar_1 \rightarrow \bullet^{K_3} \TypeVar_3)
\}$.
The set $\setE =\{ N_1 + N_2 \geq 1, N_1 + N_3 + M_1 + M_3 + K_1 + K_2 \geq 1,
M_1 + M_2 \geq 1, K_1 + K_3 \geq 1 \}
$  enforces that  $\substT_{\setC}$ is guarded, i.e.
i.e. 
$\substE \circ \substT_{\setC}$ is guarded for all $\substE$ such that
$\substE \models \setE$.
 
\end{enumerate}

\subsection{Other Auxiliary Functions}

Algorithm \ref{algorithm:test} checks whether
 a given  meta-type has certain shape and 
 Algorithm \ref{algorithm:equalcons} checks whether
 two meta-types have the same shape `at depth 0'.
  Algorithm \ref{algorithm:infinite} defines two functions
  for checking whether a meta-type is  equal to $\infinite$.
  \cref{algorithm:infinitefree} checks  whether
  the meta-type is $\infinite$-free  and 
  \cref{algorithm:tailfinite} checks whether the meta-type is tail finite.
  
 \begin{algorithm}[h]
 \SetKwInOut{Input}{input}\SetKwInOut{Output}{output}
\LinesNumbered
\SetKw{where}{{\small where}}
\SetKwHangingKw{sea}{{let}}
\SetKw{and}{{\small and}}
\SetKw{Return}{{\small return}}
\SetKw{Disj}{{\small or}}
 \Fn{$\isNatOpBox \MetaType$}{
\Return{
 $\MetaType =\nat $
 \Disj{ $\MetaType = \MetaTypeS_1 \binaryconstructor \MetaTypeS_2$  } 
 }}

  \caption{Checking the shape of the meta-type}
\label{algorithm:test}
 \end{algorithm}

\begin{algorithm}[h]

 \SetKwInOut{Input}{input}\SetKwInOut{Output}{output}
\LinesNumbered
\SetKw{where}{{\small where}}
\SetKwHangingKw{sea}{{let}}
\SetKw{and}{{\small and}}
\SetKw{Return}{{\small return}}
\SetKw{Disj}{{\small or}}

 \Fn{$\equalcons {\MetaType}{\MetaTypeS}$}{
\Return{
 $\MetaType =\nat $ \and{} $\MetaTypeS =\nat $ 
 \Disj{}
 $\MetaType = \MetaType_1 \binaryconstructor \MetaType_2$ \and{ }
 $\MetaTypeS = \MetaTypeS_1 \binaryconstructor \MetaTypeS_2$ 
 }}
 
  \caption{Checking that meta-types are equally constructed}
\label{algorithm:equalcons}
 \end{algorithm}

 \begin{algorithm}[h]
 \SetKwInOut{Input}{input}\SetKwInOut{Output}{output}
\SetKw{where}{{\small where}}
\SetKwHangingKw{sea}{{let}}
\SetKw{and}{{\small and}}
\SetKw{Return}{{\small return}}
\SetKw{Disj}{{\small or}}

\Fn{$\equalinfinite \MetaType$}{
\Input{$\MetaType$ regular meta-type}
\Output{$\equalinfinite \MetaType$ is true iff $\MetaType$ is different from 
$\infinite$.}

\If{
$\MetaType = \bullet^{\IntExp} \MetaTypeS$ \and{}
($\MetaTypeS =  \TypeVar$ \Disj{}
 $\MetaTypeS =  \nat$ \Disj{}
$\MetaTypeS =  (\Type_1 \arrow \Type_2)$ \Disj{}
$\MetaTypeS =(\Type_1 \times \Type_2)$ 
}{\Return{\True}}\lElse{\Return{\False}}
}

  \caption{Checking that a meta-type is equal to $\infinite$}
\label{algorithm:infinite}
\end{algorithm}

\begin{algorithm}[h]
 \SetKwInOut{Input}{input}\SetKwInOut{Output}{output}

\SetKwHangingKw{sea}{{let}}
\SetKw{and}{{\small and}}
\SetKw{Return}{{\small return}}
\SetKw{Disj}{{\small or}}

\Fn{$\freeinfinite \MetaType$}{
\Input{$\MetaType$ regular meta-type}
\Output{$\freeinfinite \MetaType$ is true iff $\MetaType$ is
$\infinite$-free.}

\ForEach{$\MetaTypeS \in \Subtrees{\MetaType}$}
        {\lIf{ $\neg \equalinfinite{\MetaTypeS}$}{\Return{\False}}}
\Return{\True}}

  \caption{Checking that a meta-type is $\infinite$-free}
\label{algorithm:infinitefree}
\end{algorithm}

\begin{algorithm}[H]
 \SetKwInOut{Input}{input}\SetKwInOut{Output}{output}

\SetKwHangingKw{sea}{{let}}
\SetKw{and}{{\small and}}
\SetKw{Return}{{\small return}}
\SetKw{Disj}{{\small or}}

\Fn{$\tailfinite \MetaType$}{
\Input{$\MetaType$ regular meta-type}
\Output{$\tailfinite \MetaType$ is true iff $\MetaType$ is
tail finite.}

\lIf{$\neg\freeinfinite{\MetaType}$}{\Return{\False}}
{
 $\assign{\setE}{\emptyset}$\;
 
\ForEach{$\MetaTypeS \in \Psubtrees \MetaType $ such that $\MetaTypeS \in \Psubtrees \MetaTypeS $ 
\tcc*[f]{$\MetaTypeS = \ldots \MetaTypeS \ldots $}}
{\lIf{$\Finitealt{\MetaTypeS}{\MetaTypeS}$}{\Return{$\False$}}
}
     
\Return{$\True$}
}
}

\Fn{$\Finitealt{\MetaTypeS}{\MetaType}$}{
\Input{$\MetaTypeS$, $ \MetaType$ regular  pseudo meta-types}
\Output{$\Finitealt{\MetaTypeS}{\MetaType}$ true iff
$\MetaType = \bullet^{\IntExp_{1}} (\MetaType_1 \arrow \ldots (\MetaType_n \arrow \bullet^{\IntExp_{n}} \MetaTypeS) \ldots)$
}

\BlankLine

\lIf{$\MetaType = \bullet^{\IntExp} (\MetaType_{1} \arrow \MetaType_{2})$ and
$\MetaTypeS \in \Psubtrees {\MetaType_2}$}
{\Return{$ \Finitealt{\MetaTypeS}{\MetaType_{2}}$}}
\lElse{\Return{$\False$}}

}

  \caption{Checking that a meta-type is tail finite}
\label{algorithm:tailfinite}
\end{algorithm}

\end{document}
